\documentclass[11pt,draftcls,onecolumn]{IEEEtran}
\usepackage[boxruled]{algorithm2e}
\usepackage{cite}
\usepackage{graphicx}
\usepackage{psfrag}
\usepackage{subfigure}
\usepackage{url}
\usepackage{amsmath}
\usepackage{array}
\usepackage{amssymb}
\usepackage{amsfonts}
\usepackage{graphicx}
\usepackage{epstopdf}
\usepackage{algorithmic}
\newtheorem{proposition}{Proposition}
\newtheorem{lemma}{Lemma}

\newtheorem{definition}{Definition}

\newtheorem{example}{Example}
\newtheorem{construction}{Construction}
\newtheorem{note}{Note}
\usepackage{setspace}
\usepackage{amscd}
\usepackage{mathrsfs}
\usepackage{epsfig}
\usepackage{color}
\usepackage{textcomp}
\usepackage{multirow}
\usepackage{threeparttable}

\def\c{{\mathbb{C}}}

\title{Distributed Space Time Coding for Wireless Two-way Relaying}
\begin{document}

\author{Vijayvaradharaj T. Muralidharan and B. Sundar Rajan,~\IEEEmembership{Senior Member,~IEEE}\footnote{\hrule \vspace{0.2cm} The authors are with the Dept. of Electrical Communication Engineering, Indian Institute of Science, Bangalore-560012, India (e-mail:tmvijay@ece.iisc.ernet.in; bsrajan@ece.iisc.ernet.in).}}

\maketitle
\begin{abstract}
We consider the wireless two-way relay channel, in which two-way data transfer takes place between the end nodes with the help of a relay. For the Denoise-And-Forward (DNF) protocol, it was shown by Koike-Akino et. al. that adaptively changing the network coding map used at the relay greatly reduces the impact of Multiple Access interference at the relay. The harmful effect of the deep channel fade conditions can be effectively mitigated by proper choice of these network coding maps at the relay.  Alternatively, in this paper we propose a Distributed Space Time Coding (DSTC) scheme, which effectively removes most of the deep fade channel conditions at the transmitting nodes itself without any CSIT and without any need to adaptively change the network coding map used at the relay. It is shown that the deep fades occur when the channel fade coefficient vector falls in a finite number of vector subspaces of $\mathbb{C}^2$, which are referred to as the singular fade subspaces. DSTC design criterion referred to as the \textit{singularity minimization criterion} under which the number of such vector subspaces are minimized is obtained. Also, a criterion to maximize the coding gain of the DSTC is obtained. Explicit low decoding complexity DSTC designs which satisfy the singularity minimization criterion and maximize the coding gain for QAM and PSK signal sets are provided. Simulation results show that at high Signal to Noise Ratio, the DSTC scheme provides large gains when compared to the conventional Exclusive OR network code and performs slightly better than the adaptive network coding scheme proposed by Koike-Akino et. al.
\end{abstract}

\section{Background and Preliminaries}
\subsection{Background}
We consider the two-way wireless relaying scenario shown in Fig.\ref{DNF_protocol}. Two-way data transfer takes place between the nodes A and B with the help of the relay R. It is assumed that all the three nodes operate in half-duplex mode, i.e., they cannot transmit and receive simultaneously in the same frequency band. The idea of physical layer network coding for the two way relay channel was first introduced in \cite{ZLL}, where the multiple access interference occurring at the relay was exploited so that the communication between the end nodes can be done using a two phase protocol. A protocol called Denoise-And-Forward (DNF) was proposed in \cite{PoYo_DNF}, which consists of the following two phases: the \textit{multiple access} (MA) phase (Fig. \ref{DNF_MAC}), during which A and B simultaneously transmit to R and the \textit{broadcast} (BC) phase (Fig. \ref{DNF_BC}) during which R transmits to A and B. Network coding map, which is also referred to as the denoising map, is chosen at R in such a way that A (B) can decode the messages of B (A), given that A (B) knows its own messages. During the MA phase, the transmissions from the end nodes were allowed to interfere at R, but the harmful effect of this interference was mitigated by a proper choice of the network coding map used at R. Information theoretic studies for the physical layer network coding scenario were reported in \cite{KMT},\cite{PoY}. A differential modulation scheme with analog network coding for bi-directional relaying was proposed in \cite{LiYoAnBiAt}. The design principles governing the choice of modulation schemes to be used at the nodes for uncoded transmission were studied in \cite{APT1}. An extension for the case when the nodes use convolutional codes was done in \cite{APT2}. A multi-level coding scheme for the two-way relaying scenario was proposed in \cite{HeN}. Power allocation strategies and lattice based coding schemes for bi-directional relaying were proposed in \cite{WiNa}.

%

%
%
%
  It was observed in \cite{APT1} that the network coding map used at the relay needs to be changed adaptively according to the channel fade coefficients, in order to minimize the impact of the Multiple Access Interference (MAI).  A computer search algorithm called the \textit{Closest-Neighbour Clustering} (CNC) algorithm was proposed in \cite{APT1} to obtain the adaptive network coding maps resulting in the best distance profile at R. An adaptive network coding scheme for MIMO two-way relaying based on the CNC algorithm was proposed in \cite{Ak_MIMO}. An alternative procedure to obtain the adaptive network coding maps, based on the removal of deep channel fade conditions using Latin Squares was proposed in \cite{NVR}. A quantization of the set of all possible channel realizations based on the network code used was obtained analytically in \cite{VNR}. An extension of the adaptive network coding scheme for MIMO two-way relaying using Latin Rectangles was made in \cite{VvR_MIMO}.

As an alternative to the adaptive network coding schemes in \cite{APT1} and \cite{NVR}--\cite{VNR}, in this paper, we propose a Distributed Space Time Coding (DSTC) scheme, which mitigates the effect of MAI to the fullest extent possible at the transmitting nodes itself without any CSIT. For the proposed DSTC scheme the network coding map used at R need not be changed adaptively according to channel conditions which reduces the complexity at R to a great extent and also eliminates the need for overhead bits from R to A and B to indicate the choice of the network coding map.

A distributed space time coding scheme for a wireless two-way relay network with multiple relay nodes was proposed in \cite{TaFeTrAr}, in which the DSTC was constructed at the relay nodes. In the proposed scheme, the DSTC is constructed at the end nodes A and B.

\subsection{Signal Model}
Throughout, a quasi-static fading scenario is assumed with the Channel State Information (CSI) available only at the receivers. Let $h_A$ and $h_B$ denote the fade coefficients associated with A-R and B-R links and $h'_A$ and $h'_B$ denote the fade coefficients associated with R-A and R-B links. All the fading coefficients are assumed to follow Rician distribution.

Let $\mathcal{S}$ denote the unit energy $M=2^{\lambda}$ point constellation used at the end nodes. Let $\mu: \mathbb{F}_2^\lambda \rightarrow \mathcal{S}$ denote the mapping from bits to complex symbols used at A and B.
\subsubsection{Denoise-And-Forward (DNF) protocol}
In the sequel, we briefly describe the adaptive network coding schemes based on the DNF protocol proposed in \cite{APT1}, {\cite{NVR} -- \cite{VNR}}. Throughout the paper, by DNF protocol, we refer to the schemes proposed in \cite{APT1} and \cite{NVR}--\cite{VNR}. 

In the DNF protocol, transmission occurs in two phases: Multiple Access (MA) phase during which A and B simultaneously transmit to R and Broadcast (BC) phase during which R transmits to A and B.
\subsubsection*{MA Phase}
  Let $x_A= \mu(s_A)$, $x_B=\mu(s_B)$ $\in \mathcal{S}$ denote the complex symbols transmitted by A and B respectively, where $s_A,s_B \in \mathbb{F}_2^\lambda$. The received signal at $R$ is given by,
$$y_R=h_{A} x_A + h_{B} x_B +z_R.$$ The additive noise $z_R$ is assumed to be $\mathcal{CN}(0,\sigma^2),$ where $\mathcal{CN}(0,\sigma^2)$ denotes the circularly symmetric complex Gaussian random variable with mean zero and variance $\sigma ^2.$   
\subsubsection*{BC Phase}
Let $(\hat{x}_A,\hat{x}_B) \in \mathcal{S}^2$ denote the Maximum Likelihood (ML) estimate of $({x}_A,{x}_B)$ at R based on the received complex number $y_{R}.$
Depending on the value of $h_A$ and $h_B,$ R chooses a many-to-one map $\mathcal{M}^{h_A,h_B}:\mathcal{S}^2 \rightarrow \mathcal{S}',$ where $\mathcal{S}'$ is the signal set (of size between $M$ and $M^2$) used by R during the $BC$ phase.

In order to ensure that A (B) is able to decode B's (A's) message, the map $\mathcal{M}^{h_A,h_B}$ should satisfy the exclusive law \cite{APT1}, i.e.,

{\vspace{-.5 cm}
\begin{align}
\left.
\begin{array}{ll}
\nonumber
\mathcal{M}^{h_A,h_B}(x_A,x_B) \neq \mathcal{M}^{h_A,h_B}(x'_A,x_B), \; \mathrm{for} \;x_A \neq x'_A, \; \forall \;x_B \in  \mathcal{S},\\
\nonumber
\mathcal{M}^{h_A,h_B}(x_A,x_B) \neq \mathcal{M}^{h_A,h_B}(x_A,x'_B), \; \mathrm{for} \;x_B \neq x'_B, \; \forall \;x_A \in \mathcal{S}.
\end {array}
\right\} 
\end{align}\vspace{-.5 cm}}
 
The CNC algorithm proposed in \cite{APT1} obtains the map $\mathcal{M}^{h_A,h_B}$ which results in the best distance profile during the MA phase at R. The CNC algorithm is run for all possible channel realizations and a partition of the set of all channel realizations is obtained depending on the chosen network coding map. For a given channel realization, the choice of the network coding map is indicated to A and B using overhead bits. During the BC phase R transmits $x_R=\mathcal{M}^{h_A,h_B}(\hat{x}_A,\hat{x}_B) \in \mathcal{S'}.$  The received signals at A and B during the BC phase are respectively given by,

{\vspace{-.5 cm}
\begin{align*}
y_A=h'_{A} x_R +
 z_A,\\y_B=h'_{B} x_R + z_B,
\end{align*}
\vspace{-.6 cm}}

\hspace{-.4 cm}where $z_A$ and $z_B$ are independent and $\mathcal{CN}(0, \sigma^2).$
Since the map $\mathcal{M}^{h_A,h_B}$ satisfies the exclusive law and A (B) knows its own message $x_A$ ($x_B$), it can decode $x_B$ ($x_A$) by decoding $x_R.$ 

 The CNC algorithm optimizes the entire distance profile instead of maximizing only the minimum distance. In some cases, this results in the use of signal sets with a larger cardinality during the BC phase. To solve this problem, an algorithm called the Nearest Neighbour Clustering (NNC)  algorithm was proposed in \cite{APT1} which maximizes the minimum distance alone, instead of optimizing the entire distance profile. 

The choice of the network coding map obtained depends only on the ratio $\frac{h_B}{h_A}$ and not the individual values of $h_A$ and $h_B$ \cite{APT1}.
In \cite{NVR}, the values of $\frac{h_B}{h_A}$ for which deep channel conditions occur were identified and network coding maps which remove the harmful effect of these  deep channel conditions were obtained by the completion of partially filled Latin Squares. 
\subsubsection{The Proposed DSTC Scheme}  
For the proposed DSTC scheme, transmission occurs in four phases: Two MA phases during which A and B simultaneously transmit to R followed by two BC phases during which R transmits to A and B. Two independent complex symbols each from A to B and B to A get exchanged at the end of the four phases and hence the information rate in bits per channel use for the proposed scheme is same as that of the DNF protocol.  

\subsubsection*{MA Phases}
Let $x_{A_1}= \mu(s_{A_1}),x_{A_2}= \mu(s_{A_2}) \in \mathcal{S}$  denote two independent complex symbols A wants to communicate to B. Similarly, B wants to communicate two independent complex symbols $x_{B_1}=\mu(s_{B_1}), x_{B_2}=\mu(s_{B_2}) \in \mathcal{S}$ to A. During the $i^{th}$ MA phase $i \in \lbrace 1,2 \rbrace,$ A transmits $f_A^i(x_{A_1},x_{A_2}) \in \mathbb{C},$ a function  of $x_{A_1}$ and $x_{A_2},$ and similarly B transmits $f_B^i(x_{B_1},x_{B_2}) \in \mathbb{C},$ a function of $x_{B_1}$ and $x_{B_2}.$ The received signal at R during the two MA phases can be written as,

{\vspace{-.5 cm}
\begin{align*}
\mathbf{y_R}=\left[ y_{R_1} \: y_{R_2} \right] &=  \left[ h_A \: h_B  \right] \begin{bmatrix} f_A^1(x_{A_1},x_{A_2}) & f_A^2(x_{A_1},x_{A_2}) \\  f_B^1(x_{B_1},x_{B_2}) & f_B^2(x_{B_1},x_{B_2}) \end{bmatrix} + \left[ \begin{array}{c} z_{R_1}  z_{R_2} \end{array} \right],\end{align*}}where $y_{R_i}$ denotes the received signal at R during the $i^{th}$ MA phase, $z_{R_1}$ and $z_{R_2}$ are independent and  $\mathcal{CN}(0, \sigma ^2).$ Let $\mathbf{x_A}=[x_{A_1} x_{A_2}]$ and $\mathbf{x_B}=[x_{B_1} x_{B_2}].$
The matrix, {\begin{equation}
\label{DSTC_design}
\mathbf{C}(\mathbf{x_A},\mathbf{x_B})=\begin{bmatrix} f_A^1(x_{A_1},x_{A_2}) & f_A^2(x_{A_1},x_{A_2}) \\  f_B^1(x_{B_1},x_{B_2}) & f_B^2(x_{B_1},x_{B_2}) \end{bmatrix}
\end{equation}}represents a \textit{DSTC codeword matrix}. Note that in the DSTC codeword matrix, $x_{A_1}$ and $x_{A_2}$ can occur only in the first row and, $x_{B_1}$ and $x_{B_2}$ can occur only in the second row. In this way the DSTC differs from space time codes for the conventional 2$\times$1 multiple antenna system with two collocated antennas at the transmitter in which the complex symbols can occupy any entry in the codeword matrix. 

For a complex number $x,$ let $x^R$ and $x^I$ denote the real and imaginary parts of $x.$

\begin{definition}
A DSTC is said to be linear if the entries of the first row of the codeword matrices are complex linear combinations of $x_{A_1} ^R,x_{A_1}^ I,x_{A_2} ^R,x_{A_2} ^I$ and the entries of the second row are complex linear combinations of $x_{B_1} ^R,x_{B_1}^ I,x_{B_2} ^R,x_{B_2} ^I.$ Any codeword matrix $\mathbf{C}(\mathbf{x_A},\mathbf{x_B})$ of a linear DSTC can be written as,

{
\vspace{-.4 cm}
\begin{equation} 
\label{DSTC_weight}
\mathbf{C}(\mathbf{x_A},\mathbf{x_B})=\sum_{i=1,2} \mathbf{W_{A_i}^R} x_{A_i}^R+\mathbf{W_{A_i}^I} x_{A_i}^I+\mathbf{W_{B_i}^R} x_{B_i}^R+\mathbf{W_{B_i}^I} x_{B_i}^I.
\vspace{-.4 cm}
\end{equation}}
\end{definition}

The matrices  $\mathbf{W_{A_i}^R, W_{A_i}^I, W_{B_i}^R}$ and $ \mathbf{W_{B_i}^I}$ are referred to as the \textit{weight matrices} of the DSTC. Note that the entries of the second (first) row are zeros in the matrices $\mathbf{W_{A_i}^R}$ and $\mathbf{W_{A_i}^I}$ ($\mathbf{W_{B_i}^R}$ and $\mathbf{W_{B_i}^I}$).   
\begin{definition}
A linear DSTC is said to be over the signal set $\mathcal{S}$ if the entries of the first (second) row of the codeword matrices are complex linear combinations of $x_{A_1}$ and $x_{A_2}$ ($x_{B_1}$ and $x_{B_2}$), where $x_{A_1},$ $x_{A_2},$ $x_{B_1}$ and $x_{B_2}$ belong to the signal set $\mathcal{S}.$
\end{definition} 
For a linear DSTC over $\mathcal{S}$, codeword matrix $\mathbf{C}(\mathbf{x_A},\mathbf{x_B})$ is of the form  $\mathbf{C}(\mathbf{x_A},\mathbf{x_B})=\begin{bmatrix} \mathbf{x_A} \mathbf{M_A} \\ \mathbf{x_B} \mathbf{M_B} \end{bmatrix},$ where $\mathbf{M_A}$ and $\mathbf{M_B}$ are $2 \times 2$ complex matrices referred to as the \textit{generator matrices} at node A and B respectively. Throughout the paper, we consider only linear DSTCs over a signal set $\mathcal{S}.$

\subsubsection*{BC Phases}
 Let {\small$(\hat{s}_{A_1},\hat{s}_{A_2},\hat{s}_{B_1},\hat{s}_{B_2})$} denote the maximum likelihood estimate of {\small$({s}_{A_1},{s}_{A_2},{s}_{B_1},{s}_{B_2})$} at R. The relay R transmits $x_{R_1}=\mu(\hat{s}_{A_1} \oplus \hat{s}_{B_1})$ and $x_{R_2}=\mu(\hat{s}_{A_2} \oplus \hat{s}_{B_2})$ during the first and second BC phases respectively, where $\oplus$ denotes the bit-wise XOR operation. The received signals at the end nodes during the two BC phases are given by, $y_{A_i}=h'_A x_{R_i}+z_{A_i}$ and $y_{B_i}=h'_B x_{R_i}+z_{B_i},$ where $i \in \lbrace 1,2 \rbrace.$ Since A (B) knows its own messages and the XOR map satisfies the exclusive law, A (B) can decode $s_{B_i}$ ($s_{A_{i}}$) $i \in \lbrace 1,2 \rbrace,$ by decoding $x_{R_i}.$

 Note that for the proposed DSTC scheme the signal set used during the BC phase is of the minimum cardinality $2^{\lambda}$ (the cardinality of the signal set should be at least $2^{\lambda}$ to convey $\lambda$ information bits). In contrast, for the scheme proposed in \cite{APT1}, depending on channel conditions unconventional signal sets with cardinality greater than the minimum cardinality are required. Minimum cardinality signal set is used during the BC phase and throughout the paper the focus is on optimizing the performance during the MA phase.

Some of the advantages of the proposed DSTC scheme over the schemes proposed in \cite{APT1}, \cite{NVR}--\cite{VNR} are summarized below:

\begin{itemize}
\item
Unlike the schemes proposed in \cite{APT1},\cite{NVR}--\cite{VNR}, for the proposed DSTC scheme, the network coding map used at R need not be changed adaptively according to channel conditions. Any network coding map satisfying the exclusive law will give the same performance and for simplicity, the conventional bit-wise Exclusive OR (XOR) map itself can be used. This reduces the complexity at R to a great extent and also eliminates the need for overhead bits from R to A and B to indicate the choice of the network coding map.
\item
For the scheme proposed in \cite{APT1}, for certain channel conditions the adaptive network coding map necessitates the use of unconventional signal sets with cardinality greater than the minimum cardinality required during the BC phase, which results in a degradation in performance. For the proposed scheme, the relay always uses a conventional signal set with minimum cardinality.  
\item
The adaptive network coding maps were obtained in \cite{APT1}, by exhaustive computer search. For the proposed scheme no such computer search is required, since the same network code is used irrespective of channel conditions.    
\end{itemize} 
 
 The contributions and organization of the paper are as follows:
 
\begin{itemize}
\item
For a classical $n_t \times n_r$ MIMO system with collocated antennas, deep channel fade conditions occur when the channel fade coefficient vector belongs to a finite number of vector subspaces of $\mathbb{C}^{n_t}$ referred to as the singular fade subspaces. The way in which transmit diversity schemes (space time codes) remove the harmful effect of these singular fade subspaces is discussed. The connection between the dimension of these singular fade subspaces and the transmit diversity order is explained (Section II).
\item 
The MAC phase of the DNF protocol for the two-way relaying scenario can be viewed as a virtual $2 \times 1$ MISO system. The singular fade subspaces for the classical $2 \times 1$ MISO system, are singular fade subspaces for the two-way relaying scenario as well. The connection between dimension of these singular fade subspaces and the diversity order for the adaptive network coding schemes proposed in \cite{APT1} and \cite{NVR}-\cite{VNR} is discussed (Section III A).
\item 
  The singular fade subspaces for the proposed DSTC scheme are identified. The goal of minimizing the number of singular fade subspaces results in a new design criterion referred as the \textit{singularity minimization criterion} for  DSTCs. It is shown that for a properly chosen DSTC, most of the vector subspaces which were singular fade subspaces for the DNF protocol, are no longer singular fade subspaces for the DSTC scheme. Also, a criterion to maximize the coding gain of the proposed DSTC scheme is obtained (Section III B).
\item
It is shown that for DSTCs which are over $\mathcal{S},$ where $\mathcal{S}$ is a square QAM or $2^{\lambda}$-PSK signal set, the coding gain is maximized when the generator matrices $\mathbf{M_A}$ and $\mathbf{M_B}$ at nodes A and B are unitary matrices. Explicit construction of DSTCs over QAM and PSK signal sets which satisfy the singularity minimization criterion and maximize the coding gain are provided. It is shown that for all DSTCs over $\mathcal{S}$ with unitary generator matrices $\mathbf{M_A}$ and $\mathbf{M_B}$, the ML decoding complexity at R is $O(M^3)$ for any arbitrary signal set and is $O(M^2)$ for square QAM signal sets. Note that the brute force ML decoding complexity is $O(M^4)$ (Section IV).
\item 
 Simulation results presented in Section V show that at high SNR, the DSTC scheme provides large gains when compared to the conventional XOR network code based on the DNF protocol and performs slightly better than the adaptive network coding scheme proposed in \cite{APT1}. 
 \end{itemize}

\textbf{\textit{Notations}:}
  The complex number $\sqrt{-1}$ is denoted by $j.$ The set of integers, Gaussian integers, rational, real and complex numbers are respectively denoted as $\mathbb{Z},\mathbb{Z}[j], \mathbb{Q}, \mathbb{R}$ and $\mathbb{C}.$ All the vector spaces and vector subspaces considered in this paper are over the complex field $\mathbb{C},$ unless explicitly mentioned otherwise. Throughout, vectors are denoted by bold lower case letters and matrices are denoted by bold capital letters. Let $\mathcal{CN}(0, \sigma ^2 \mathbf{I_n})$ denote the circularly symmetric complex Gaussian random vector with zero mean and covariance matrix $\sigma ^2 \mathbf{I_n},$ where $\mathbf{I_n}$ denotes the $n \times n$ identity matrix. Let $\langle\mathbf{c_1},\mathbf{c_2}, \dotso \mathbf{c_L}\rangle$ denote the vector subspace over $\mathbb{C}$ spanned by the complex vectors $\mathbf{c_1},\mathbf{c_2}, \dotso \mathbf{c_L}.$ For a matrix $\mathbf{A},$ $\mathbf{A^T}$ and $\mathbf{A^H}$ denotes its transpose and conjugate transpose respectively. For a vector subspace $V$ of a vector space, $V^{\perp}$ denotes the vector subspace $\lbrace \mathbf{x} : \mathbf{x^T v} =0 ,\forall \mathbf{v} \in V \rbrace$ and $\text{dim}(V)$ denotes the dimension of $V.$  The all zero vector of length $n$ is denoted by $\mathbf{0_{n}}.$  For a square matrix $\mathbf{A},$ let $\mathrm{rank}(\mathbf{A})$ denote its rank and let $\det(\mathbf{A})$ denote its determinant. For a complex number $x,$ $x^R$ and $x^I$ denote the real and imaginary parts of $x,$ $x^*$ denotes its conjugate and $\vert x \vert$ denotes its absolute value. For a vector $\mathbf{v},$ $\parallel \mathbf{v} \parallel$ denotes its Euclidean norm. For a matrix $\mathbf{A},$ $Row(\mathbf{A})$ and $Col(\mathbf{A})$ respectively denote the row space and column space of $\mathbf{A}.$ $\mathbb{E}(X)$ denotes the expectation of $X.$
\section{The Notion of Singular Fade Subspaces for the Collocated MIMO system}
\label{sing_MIMO}
In this section, to explain the notion of singular fade subspaces, we digress from the two-way relaying scenario and focus on the classical MIMO system with collocated antennas.  
Consider the classical MIMO system with $n_t$ transmit antennas at the transmitter Tx and $n_r$ receive antennas at the receiver Rx, with $\mathbf{H}$ being the $n_r \times n_t$ complex fade coefficient matrix. The entries of the matrix $\mathbf{H}$ are assumed to be i.i.d. and Rician distributed. 

\subsection{Singular Fade Subspaces for the Collocated MIMO system with Spatial Multiplexing}
Consider the spatial multiplexing of independent complex symbols at Tx, i.e., the received complex vector at Rx is given by $\mathbf{y}=\mathbf{H}\mathbf{x}+\mathbf{z},$ where $\mathbf{x}$ is the transmitted message vector of length $n_t$ whose components independently take values from the signal set $\mathcal{S}$ and $\mathbf{z}$ is $\mathcal{CN}(0, \sigma ^2 \mathbf{I_{n_t}}).$

Let $\mathcal{S}_{Rx}(\mathbf{H}) \subset \mathbb{C}^{n_r}$ denote the effective signal set at Rx, i.e.,  $\mathcal{S}_{Rx}(\mathbf{H})= \lbrace \mathbf{H}\mathbf{x} : \mathbf{x} \in \mathcal{S}^{n_t} \rbrace.$
Let ${\Delta \mathcal{S}}$ denote the difference constellation of the signal set $\mathcal{S},$ i.e., $\Delta \mathcal{S} = \lbrace s - s': s,s' \in \mathcal{S} \rbrace.$ 
The distances between two points in the effective constellation $\mathcal{S}_{Rx}(\mathbf{H})$ are of the form $\parallel \mathbf{H} \mathbf{\Delta x} \parallel,$ where $\mathbf{\Delta x} \neq \mathbf{0_{n_t}}, \mathbf{\Delta x} \in \Delta \mathcal{S} ^{n_t}.$ 

\begin{definition}
For an $n_t \times n_R$ MIMO system, the channel fade coefficient matrix $\mathbf{H}$ is said to be a $\textit{deep fade matrix}$ if the minimum distance of the effective constellation $\mathcal{S}_{Rx}(\mathbf{H})$ is zero. The row space of a deep fade matrix is said to be a \textit{deep fade space}.
\end{definition}

Let $\mathbf{h_k}, 1 \leq k \leq n_r,$ denote the $k^{th}$ row of $\mathbf{H}.$ Since $\parallel \mathbf{H} \mathbf{\Delta x} \parallel^2=\sum_{k=1}^{n_r} \vert \mathbf{h_k \Delta x}\vert^2,$ for the minimum distance of the effective constellation $\mathcal{S}_{Rx}(\mathbf{H})$ to be zero, all the vectors $\mathbf{h_k^T}, 1 \leq k \leq n_r,$ should fall in a vector subspace of the form $\langle\mathbf{\Delta x}\rangle^{\perp}$ for some $\mathbf{\Delta x} \in \Delta \mathcal{S} ^{n_t}.$ In other words, for $\parallel \mathbf{H} \mathbf{\Delta x} \parallel$ to be zero, the row space of $\mathbf{H}$ should be a subspace of the vector subspace of $\mathbb{C}^{n_t}$ of the form $\langle\mathbf{\Delta x}\rangle^{\perp}$ for some $\mathbf{\Delta x} \in \Delta \mathcal{S} ^{n_t}.$  
 The vector subspaces of the form $\langle\mathbf{\Delta x}\rangle^{\perp}$ are referred to as the \textit{singular fade subspaces}. Formally, a singular fade subspace can be defined as follows:

\begin{definition}
A vector subspace $\mathcal{V}$ of $\mathbb{C}^{n_t}$ is said to be a singular fade subspace if all the vector subspaces of $\mathcal{V}$ are deep fade spaces.
\end{definition}

Note that {\footnotesize$\left\langle\begin{bmatrix} 0 \\ 0 \end{bmatrix} \right \rangle$} is always a singular fade subspace referred to as the \textit{trivial singular fade subspace}. 
\begin{example}Consider the $2 \times 1$ MISO system with spatial multiplexing with 4-PSK signal set $\mathcal{S}=\lbrace \pm 1, \pm j \rbrace.$ The difference constellation of 4 PSK signal set has 9 points $\Delta \mathcal{S}=\lbrace 0, \pm 2, \pm 2j , \pm 1\pm j\rbrace.$ For this case, the set of fourteen singular fade subspaces, which are of the form $\langle\mathbf{\Delta x}\rangle^{\perp},$ where $\mathbf{\Delta x} \in \Delta \mathcal{S} ^{2}$ are given by, 

{\scriptsize
\begin{align} 
\nonumber
&\left \lbrace\left\langle\begin{bmatrix} 0 \\ 1 \end{bmatrix} \right \rangle, \left\langle\begin{bmatrix} 1 \\ 0 \end{bmatrix} \right \rangle, \left\langle\begin{bmatrix} 1 \\ 1 \end{bmatrix} \right \rangle,  \left\langle\begin{bmatrix} 1 \\ -1 \end{bmatrix} \right \rangle,   \left\langle\begin{bmatrix} 1 \\ j \end{bmatrix} \right \rangle,   \left\langle\begin{bmatrix} 1 \\ -j \end{bmatrix} \right\rangle,   \left\langle\begin{bmatrix} 1 \\ 1+j \end{bmatrix} \right \rangle,   \left\langle\begin{bmatrix} 1 \\ -1+j \end{bmatrix} \right \rangle, \right.\\
\label{sfs_example}
 & \left. \hspace{3 cm}   \left\langle\begin{bmatrix} 1 \\ 1-j \end{bmatrix} \right \rangle,   \left\langle\begin{bmatrix} 1 \\ -1-j \end{bmatrix} \right \rangle,   \left\langle\begin{bmatrix} 1 \\ 0.5+0.5j \end{bmatrix} \right \rangle,   \left\langle\begin{bmatrix} 1 \\ -0.5+0.5j \end{bmatrix} \right \rangle ,   \left\langle\begin{bmatrix} 1 \\ 0.5-0.5j \end{bmatrix} \right \rangle,    \left\langle\begin{bmatrix} 1 \\ -0.5-0.5j \end{bmatrix} \right \rangle \right \rbrace.
\end{align}}
The fade coefficient matrix (which is a row vector for this example) is a deep fade matrix (vector) if the row space of the fade coefficient vector is a subspace of one of these 14 singular fade subspaces, i.e., the fade coefficient vector should belong to one of these 14 vector subspaces. For example, $[2 \quad 1+j]^T$ belongs to the vector subspace {\footnotesize$\left \langle \begin{bmatrix} 1 \\ 0.5+0.5j \end{bmatrix} \right \rangle$} and is a deep fade matrix.

 
\end{example}

Note that the singular fade subspaces depend only on the number of transmit antennas $n_t$ and the signal set $\mathcal{S}.$ They are independent of the number of receive antennas $n_r,$ as illustrated in the following example.

\begin{example}
Consider the 2 $\times$ 2 MIMO system with 4-PSK signal set $\mathcal{S}=\lbrace \pm 1, \pm j \rbrace.$ The set of 14 singular fade subspaces for this case is the same as that of $2 \times 1$ MISO system given in \eqref{sfs_example}. For a fade coefficient matrix to be a deep fade matrix, both the rows should belong to one of these 14 vector subspaces. For example, $\begin{bmatrix}2 & 1+j\\ 1 & 0.5+0.5j \end{bmatrix}$ is a deep fade matrix since $\begin{bmatrix}2 & 1+j \end{bmatrix}^T$ and $\begin{bmatrix} 1 & 0.5+0.5j \end{bmatrix}^T$ belong to the vector subspace $\left \langle \begin{bmatrix} 1 \\ 0.5+0.5j \end{bmatrix} \right \rangle.$
\end{example}

The dimension of the singular fade subspace $\langle\mathbf{\Delta x}\rangle^{\perp},$  and the transmit diversity order of the pair-wise error event $(\mathbf{x} \rightarrow \mathbf{x'}),$ are inherently connected, where $\mathbf{\Delta x}= \mathbf{x}-\mathbf{x'}$ and $\mathbf{x}, \mathbf{x'} \in \mathcal{S}^{n_t}$. With spatial multiplexing, the transmit diversity order of the pair-wise error event $(\mathbf{x} \rightarrow \mathbf{x'})$ is 1 while $dim(\langle\mathbf{\Delta x}\rangle^{\perp})=n_t-1.$ It is the presence of these $n_t-1$ dimensional singular fade subspaces that results in a transmit diversity order of 1. 

The receive diversity order $n_r$ comes due to the fact that for a fade coefficient matrix to be a deep fade matrix, all the $n_r$ rows of the fade coefficient matrix should belong to the same singular fade subspace.

The use of full diversity space times space time codes results in the maximum transmit diversity order $n_t.$ In the next subsection, the connection between the singular fade subspaces of space time codes and transmit diversity order will be established.
\subsection{Singular Fade Subspaces for the Collocated MIMO system with Space Time Coding}
Consider the case when Tx uses a space time code $\mathcal{C}$ of size $n_t \times T,$ where $T \geq n_t.$ Let $\mathbf{C(x)}$ denote a codeword matrix of the space time code, where $\mathbf{x} \in \mathcal{S}^{K},$ where $K$ denotes the number of independent complex symbols transmitted. Similar to the spatial multiplexing case, the effective constellation at Rx which is a subset of $\mathbb{C} ^{n_r \times T}$ can be defined. It is easy to verify that the minimum distance of the effective constellation at Rx becomes zero when $Row(\mathbf{H})$ is a subspace of the vector subspace $Col ^{\perp}\left (\mathbf{C\left(\Delta x\right)}\right),$ for some $\mathbf{\Delta x} \in \Delta \mathcal{S}^K.$ Note that $Col ^{\perp}\left (\mathbf{C\left(\Delta x\right)}\right)$ denotes the vector subspace $\lbrace \mathbf{u} : \mathbf{u^T v}=0, \forall \mathbf{v} \in Col \left (\mathbf{C\left(\Delta x\right)}\right) \rbrace.$ The vector subspaces $Col^{\perp}\left (\mathbf{C\left(\Delta x\right)}\right)$ are the singular fade subspaces for the $n_t$ transmit antenna system with the space time code $\mathcal{C}.$

\begin{note}
Even though the probability that $Row(\mathbf{H})$ is a subspace of one of the singular fade subspaces is zero, with a non-zero probability $Row(\mathbf{H})$ falls in the neighbourhood of a subspace of one of the singular fade subspaces, which results in low values of the minimum distance of the effective constellation.
\end{note}

The dimension of the singular fade subspace $Col^{\perp}\left (\mathbf{C\left(\Delta x\right)}\right)$ is equal to $n_t-\text{rank}(\mathbf{C\left(\Delta x\right)}),$ while the transmit diversity order for the pair-wise error event $(\mathbf{x} \rightarrow \mathbf{x'}), \mathbf{x, x'} \in \mathcal{S}^{K}$ equals $\text{rank}(\mathbf{C\left(\Delta x\right)})$ \cite{TaSeCa}, where $\mathbf{\Delta x}=\mathbf{x}-\mathbf{x'}.$ With every pair-wise error event $(\mathbf{x} \rightarrow \mathbf{x'}),$ we can associate a singular fade subspace $Col^{\perp}\left (\mathbf{C\left(\Delta x\right)}\right).$ Among all the pair-wise error events, those error events for which the codeword difference matrix has the least rank determine the overall system transmit diversity order. Equivalently, among all the pair-wise error events, those error events for which the associated singular fade subspace has the largest dimension will dominate the overall error probability. This is expected since among all the singular fade subspaces, the probability that $Row(\mathbf{H})$  falls in the neighbourhood of a subspace of the singular fade subspace, will be the largest for those singular fade subspaces which have the largest dimension.

If the space time code is such that $\mathbf{C\left(\Delta x\right)}$ is full rank for all $\mathbf{\Delta x} \neq \mathbf{0_{K}},$ all the singular fade subspaces $Col^{\perp}\left (\mathbf{C\left(\Delta x\right)}\right)$ collapse to be the zero-dimensional trivial singular fade subspace {\footnotesize$\left \langle \begin{bmatrix}0 \\0  \end{bmatrix} \right\rangle,$} thereby ensuring a transmit diversity order of $n_t$ for all the pair-wise error events.

\begin{example}
Consider the $2 \times 1$ MISO system with Alamouti space time code whose design matrix is given by {\footnotesize$\begin{bmatrix}x_1 & x_2 \\ -x_2^*& x_1 ^*\end{bmatrix}.$} Since the design matrix is full rank for all choices of $x_1$ and $x_2,$ the column space of the codeword difference matrix is always $\mathbb{C}^2$ and hence all the singular fade subspaces collapse to be the zero dimensional trivial singular fade subspace {\footnotesize$\left\langle\begin{bmatrix} 0 \\ 0 \end{bmatrix}\right\rangle.$} Equivalently, all the pair-wise error events $(\mathbf{x_1,x_2}) \rightarrow (\mathbf{x'_1,x'_2})$ have a transmit diversity order 2. The full rank Alamouti space-time code removed the effect of all the vector subspaces which were non-trivial singular fade subspaces for the spatial multiplexing system, thereby increasing the diversity order of all the pair-wise error events from 1 to 2.
\end{example}

\begin{example}
For a $2^a \times 2^a$ Generalized Linear Complex Orthogonal Design (GCOD) \cite{SuXia}, the design matrix $\mathbf{G_{2^a}}(x_1,x_2,\dotso x_{a+1})$ constructed iteratively is given by,  $$\begin{bmatrix} \mathbf{G_{2^{a-1}}}(x_1,x_2,\dotso x_{a}) & x_{a+1} \mathbf{I_{2^{a-1}}} \\ -x^*_{a+1} \mathbf{I_{2^{a-1}}} & \mathbf{G^H_{2^{a-1}}}(x_1,x_2,\dotso x_{a})\end{bmatrix}.$$ The codeword difference matrix for the GCOD is full rank for any signal set. Hence, irrespective of the signal set, the trivial singular fade subspace $\left \langle \mathbf{0_{2^a}} \right \rangle$ is the only singular fade subspace for the GCOD.\\
\end{example}
\begin{example}
Consider the $4 \times 4$ Quasi-Orthogonal Design (QOD) \cite{HaJa}, whose codeword matrix is given by {$\begin{bmatrix} x_1 &-x_2^*&-x_3^*& x_4\\x_2 &x_1^*&-x_4^*& -x_3 \\ x_3 &-x_4^*&x_1^*& -x_2 \\x_4 &x_3^*&x_2^*& x_1 \end{bmatrix}.$ Let $\Delta x_i = x_i-x'_i.$} Irrespective of the signal set used, the minimum rank of the codeword difference matrix for the $4 \times 4$ QOD is 2.  For example, when $\Delta x_1= \Delta x_4 = \Delta s_1$ and $\Delta x_2= -\Delta x_3 = \Delta s_2,$ the rank of the codeword difference matrix is 2. Equivalently, there exists a non-trivial singular fade subspace, {\footnotesize $\left \langle \begin{bmatrix} \Delta s_1 & \Delta s_2 & -\Delta s_2 & \Delta s_1 \end{bmatrix}^T, \begin{bmatrix}  -\Delta s_2^* & \Delta s_1^* & -\Delta s_1^* & -\Delta s_2^*\end{bmatrix}^T\right \rangle^{\perp}.$}

\end{example}

Note that the 2 $\times 2$ Alamouti code removes the effect of the harmful non-trivial singular fade subspaces for any signal set. On the other hand, for the $4 \times 4$ QOD there exists non-trivial singular fade subspaces for any signal set.

 In general, a space time code can offer full transmit diversity for some but not all signal sets. In other words, for some signal set, a space time code might have only the trivial singular fade subspace, while for some other signal set, the same space time code might have non-trivial singular fade subspaces. For a space time code which does not offer full transmit diversity for a signal set, there would exist non-trivial singular fade subspaces. These are illustrated in the following example.

\begin{example}
\label{example_2_CIOD}
Consider the $2 \times 2$ Co-ordinate Interleaved Orthogonal Design (CIOD) \cite{KhRa} whose codeword matrices are of the form $\begin{bmatrix}x_1^R+jx_2^I & 0 \\0 & x_2^R+jx_1^I \end{bmatrix},$ where $x_1, x_2 \in \lbrace \pm 1, \pm j \rbrace.$ Let $\Delta x_i = x_i-x'_i.$ The codeword difference matrix is not full rank in the following two cases: \\
\textit{Case 1:} $\Delta x_1^R= \Delta x_2^I=0$ and at least one out of $\Delta x_1^I$ and  $\Delta x_2^R$ is non-zero. For this case, the singular fade subspace is given by $\left\langle \begin{bmatrix} 1 \\ 0 \end{bmatrix} \right\rangle.$ \\
\textit{Case 2:}
$\Delta x_1^I= \Delta x_2^R=0$ and at least one out of $\Delta x_1^R$ and $\Delta x_2^I$ is non-zero. For this case, the singular fade subspace is given by $\left\langle \begin{bmatrix} 0 \\ 1 \end{bmatrix} \right\rangle.$ \\Hence, there exists the following two non-trivial singular fade subspaces: {\footnotesize $\left\langle \begin{bmatrix} 0 \\ 1 \end{bmatrix} \right\rangle$ and $\left\langle \begin{bmatrix} 1 \\ 0 \end{bmatrix} \right\rangle.$} However, when the signal set is $e^{j \theta}\lbrace \pm 1, \pm j \rbrace,$ where $\theta$ is not a multiple of $\frac{\pi}{4},$ the $2 \times 2$ CIOD offers full transmit diversity. Equivalently, there are no singular fade subspaces other than the trivial singular fade subspace for the $2 \times 2$ CIOD with the signal set $e^{j \theta}\lbrace \pm 1, \pm j \rbrace,$ when $\theta$ is not a multiple of $\frac{\pi}{4}.$   
%
\end{example}

\begin{example}
\label{example_4_CIOD}
Consider the $4 \times 4$ CIOD \cite{KhRa} whose codeword matrices are of the form $$\begin{bmatrix}x_1^R+jx_3^I & x_2^R+jx_4^I & 0 & 0\\-x_2^R+jx_4^I & x_1^R-jx_3^I & 0 & 0 \\ 0& 0&x_3^R+jx_1^I & x_2^R+jx_4^I \\0 & 0 &-x_4^R+jx_2^I & x_3^R-jx_1^I \end{bmatrix},$$ where $x_1, x_2, x_3, x_4 \in \lbrace \pm 1, \pm j \rbrace.$ For the 4-PSK signal set considered, this STC does not offer full transmit diversity and there are pair-wise error events which have a transmit diversity order less than 2. The determinant of the codeword difference matrix for this STC is given by, 

{\vspace{-.2 cm} \footnotesize \begin{align*}&\left(\vert\Delta {x_1^R}\vert^2+\vert\Delta {x_3^I}\vert^2 +\vert\Delta {x_2^R}\vert^2+\vert\Delta {x_4^I}\vert^2 \right) \left(\vert\Delta {x_3^R}\vert^2+\vert\Delta {x_1^I}\vert^2 +\vert\Delta {x_4^R}\vert^2+\vert\Delta {x_2^I}\vert^2 \right). \end{align*}} Hence the code-word difference matrix is not full rank in the following two cases: \\
\textit{Case 1:}
 $\Delta {x_1^R}=\Delta {x_3^I}=\Delta {x_2^R}=\Delta {x_4^I}=0$ and at least one out of $\Delta {x_1^I},\Delta {x_3^R},\Delta {x_2^I},\Delta {x_4^R}$ is non-zero. For this case, the first two columns of the codeword difference matrices are zeros. The column span of the codeword difference matrix is {\footnotesize $\left \langle \begin{bmatrix} 0 &0& 1& 0 \end{bmatrix}^T,\begin{bmatrix} 0 &0& 0& 1 \end{bmatrix}^T\right\rangle$} and hence the corresponding singular fade subspace is given by {\footnotesize $\left \langle \begin{bmatrix} 1 &0& 0& 0 \end{bmatrix}^T,\begin{bmatrix} 0 &1& 0& 0 \end{bmatrix}^T\right\rangle.$} \\
 \textit{Case 2:}
 $\Delta {x_1^I}=\Delta {x_3^R}=\Delta {x_2^I}=\Delta {x_4^R}= 0$ and at least one out of $\Delta {x_1^I},\Delta {x_3^Q},\Delta {x_2^I},\Delta {x_4^Q}$ is non-zero. Similar to \textit{Case 1}, it can be shown that the singular fade subspace for this case is given by {\footnotesize $\left \langle \begin{bmatrix} 0 &0& 1& 0 \end{bmatrix}^T,\begin{bmatrix} 0 &0& 0& 1 \end{bmatrix}^T\right\rangle.$} 
 
Hence, for the $4 \times 4$ CIOD, with 4-PSK signal set $x_1, x_2, x_3, x_4 \in \lbrace \pm 1, \pm j \rbrace,$ there exists two non-trivial singular singular fade subspaces. Similar to the $2 \times 2$ CIOD, when the signal set is a rotated 4-PSK signal set, $e^{j \theta}\lbrace \pm 1, \pm j \rbrace,$ where $\theta$ is not a multiple of $\frac{\pi}{4},$ the $4 \times 4$ CIOD offers full transmit diversity and there are no non-trivial singular fade subspaces. \\
\end{example}

In general, the $2 \times 2$ CIOD given in Example \ref{example_2_CIOD} and the $4\times4$ CIOD given in Example \ref{example_4_CIOD}, offer full diversity for those signal sets for which the Co-ordinate Product Distance (CPD) \footnote{The CPD between two complex numbers $x$ and $y$ is defined to be $\vert x^R-y^R\vert 
\vert x^I - y^I \vert.$ The CPD of a signal set is defined to minimum among all CPDs between pairs of points in the signal set \cite{KhRa}.} is non-zero \cite{KhRa}. Equivalently, there are non-trivial singular fade subspaces for the $2 \times 2$ and $4 \times 4$ CIOD, for signal sets whose CPD is non-zero. In fact, this is true for any Generalized Co-ordinate Interleaved Orthogonal Design (GCIOD), as illustrated in the next example. \\
\begin{example}
Consider the $2^a \times 2^a$ Generalized Co-ordinate Interleaved Orthogonal Design (GCIOD) \cite{KhRa} whose codeword design matrix is given by, {\footnotesize $\begin{bmatrix} G_{2^{a-1}}(\tilde{x}_1,\dotso,\tilde{x}_a) & 0 \\0 & G_{2^{a-1}}(\tilde{x}_{a+1},\dotso,\tilde{x}_{2a}) \end{bmatrix}.$} The complex number $\tilde{x}_i=x_i^R+j x^I_{(i+a)_{2a}},$ where $(r)_{s}$ denotes $r$ modulo $s$ and $G_{2^{a-1}}({x}_1,\dotso,{x}_a)$ is the codeword matrix of the GCOD \cite{SuXia} of size $2^{a-1}.$ The determinant of the codeword difference matrix is given by {\footnotesize$\left(\sum_{i=1}^a(\vert \Delta x^R_{i} \vert^2+\vert \Delta x^I_{(a+i)_{2a}} \vert ^2\right)\left(\sum_{i=1}^a(\vert \Delta x^I_{i} \vert^2+\vert \Delta x^R_{(a+i)_{2a}} \vert ^2\right).$} The determinant is non-zero for those signal sets for which the CPD is non-zero and there are no non-trivial singular fade subspaces. For those signal sets for which the CPD is zero, the determinant becomes zero under the following two cases:

\textit{Case 1:}
$\Delta x^I_{i} =\Delta x^R_{(a+i)_{2a}}=0, \forall i \in \lbrace 1, \dotso,a \rbrace$ and at least one of the elements of the set {\footnotesize$\lbrace \Delta x^R_{i},\Delta x^I_{(a+i)_{2a}},$ $1 \leq i \leq a \rbrace$} is non-zero. It can be verified that the singular fade subspace for this case is given by $\left \langle e_1,e_2,\dotso,e_{2^{a-1}} \right \rangle,$ where $e_i$ denotes the $2^a$ length vector whose $i^{th}$ component is one and all other components are zeros. \\
\textit{Case2:}
$\Delta x^R_{i} =\Delta x^I_{(a+i)_{2a}}=0, \forall i \in \lbrace 1, \dotso,a \rbrace$ and at least one of the elements of the set {\footnotesize$\lbrace \Delta x^I_{i},\Delta x^R_{(a+i)_{2a}},$  $1 \leq i \leq a \rbrace$} is non-zero. For this case, the singular fade subspace is given by $\left \langle e_{2^{a-1}+1},\dotso,e_{2^{a}} \right \rangle.$
\end{example}

\section{Singular fade subspaces for the two-way relaying scenario}
In the previous subsection, the notion of singular fade subspaces was introduced and its connection to the transmit diversity order of the MIMO system with collocated antennas was established. Since the MA phase of the two-way relaying scenario can be viewed as a virtual $2 \times 1$ MISO system, there exists singular fade subspaces for this case as well. 

In Subsection \ref{sing_DNF}, the singular fade subspaces for the two-way relaying scenario are identified. The reason why the adaptive network coding schemes based on the DNF protocol proposed in \cite{APT1} and \cite{NVR}- \cite{VNR} mitigate the effect of these harmful singular fade subspaces is discussed. In Subsection \ref{sing_DSTC}, it is shown that minimizing the harmful effect of these singular fade subspaces can also be achieved by a proper choice of the DSTC, without any need to adaptively change the network code at R according to channel conditions.
\subsection{Singular Fade Subspaces for the DNF Protocol}
\label{sing_DNF}
%

  Let $\Delta x_A=x_A-x'_A$ and $\Delta x_B= x_B-x'_B \in \Delta \mathcal{S}.$ From the discussion in Section \ref{sing_MIMO}, it follows that the singular fade subspaces for the DNF protocol are of the form {\footnotesize $\left\langle\begin{bmatrix}\Delta x_A \\ \Delta x_B \end{bmatrix}\right\rangle ^{\perp}=\left\langle\begin{bmatrix} 1 \\ \frac{-\Delta x_A}{ \Delta x_B} \end{bmatrix}\right\rangle.$} 
The ratio $\frac{-\Delta x_A}{ \Delta x_B}$ determines all the singular fade subspaces for the DNF protocol. In \cite{NVR}-\cite{VNR}, the ratio $\frac{-\Delta x_A}{ \Delta x_B}$ was called the \textit{singular fade state}.
 
As mentioned earlier in Section \ref{sing_MIMO}, {\footnotesize $\mathrm{dim}\left(\left\langle\begin{bmatrix} 1 \\ \frac{-\Delta x_A}{ \Delta x_B} \end{bmatrix}\right\rangle\right)$} and the diversity order for the pair-wise error event that a pair $(x_A,x_B)$ is wrongly decoded at R as $(x'_A,x'_B)$ (denoted as $(x_A,x_B) \rightarrow (x'_A,x'_B)$) are inherently connected. The diversity order for the error event $(x_A,x_B) \rightarrow (x'_A,x'_B)$ is equal to $\text{rank}([\Delta x_A \: \Delta x_B])=1$ while {\footnotesize$\mathrm{dim}\left(\left\langle\begin{bmatrix} 1 \\ \frac{-\Delta x_A}{ \Delta x_B} \end{bmatrix}\right\rangle\right)=2-\text{rank}([\Delta x_A \: \Delta x_B])=1.$} 

Let $\mathcal{S}_R(h_A,h_B)=\lbrace h_A \tilde{x}_A +h_B \tilde{x}_B: \tilde{x}_A,\tilde{x}_B \in \mathcal{S}\rbrace$ denote the effective constellation at R. Let $d_{min}(h_A,h_B)$ denote the minimum distance of $\mathcal{S}_R(h_A,h_B).$ When $[h_A \; h_B]^T$ falls in one of the singular fade subspaces, $d_{min}(h_A,h_B)$ becomes zero.
Even though the  probability that the vector $[h_A \: h_B]^T$ belongs to a singular fade subspace is zero, $d_{min}(h_A,h_B)$ is greatly reduced when $[h_A \: h_B]^T$ falls close to a singular fade subspace, a phenomenon referred as \textit{distance shortening}. For $\Delta x_A \neq 0$ and $\Delta x_B \neq 0,$ the CNC algorithm \cite{APT1} avoids the distance shortening occurring in the neighbourhood of a singular fade subspace {\footnotesize$\left\langle\begin{bmatrix}\Delta x_A \\ \Delta x_B \end{bmatrix}\right\rangle ^{\perp},$} by ensuring that $\mu^{h_A,h_B}(x_A,x_B)=\mu^{h_A,h_B}(x'_A,x'_B),$ i.e., R does not distinguish the pairs $(x_A,x_B)$ and $(x'_A,x'_B),$ which are said to be clustered together. In fact, for every realization of $[h_A h_B]$ (not necessarily in the neighbourhood of singular fade subspaces), the CNC algorithm chooses the network coding map which results in the best distance profile at R by appropriate clustering of the signal points. The scheme proposed in \cite{NVR}-\cite{VNR} avoids distance shortening in the neighbourhood of singular fade subspaces by proper choice of clustering for only the singular fade subspaces and not for every realization of the channel fade coefficients. 

Consider the two singular fade subspaces: {\footnotesize$\left\langle\begin{bmatrix}0\\ \Delta x_B \end{bmatrix}\right\rangle ^{\perp}=\left\langle\begin{bmatrix} 1 \\ 0 \end{bmatrix}\right\rangle$} and {\footnotesize$\left\langle\begin{bmatrix}\Delta x_A \\ 0 \end{bmatrix}\right\rangle ^{\perp}=\left\langle\begin{bmatrix} 0 \\ 1 \end{bmatrix}\right\rangle.$} The distance shortening which occurs in the neighbourhood of these singular fade subspaces is unavoidable, since the pairs $(x_A,x_B)$ and $(x_A,x'_B)$ (and also the pairs $(x_A,x_B)$ and $(x'_A,x_B)$) which result in these singular fade subspaces cannot be clustered together without violating the exclusive law. Such singular fade subspaces are referred as the \textit{non-removable singular fade subspaces}. The dimension of these singular fade subspaces is one or equivalently, the error events $(x_A,x_B) \rightarrow (x_A,x'_B)$ and $(x_A,x_B) \rightarrow (x'_A,x_B)$ always result in diversity order one. The singular fade subspaces other than the non-removable singular fade subspaces are referred as the \textit{removable singular fade subspaces}. 

The removable singular fade subspaces are of the form {\footnotesize$\left\langle\begin{bmatrix} 1 \\ \frac{-\Delta x_A}{ \Delta x_B} \end{bmatrix}\right\rangle, \Delta x_A \neq 0, \Delta x_B \neq 0,$} which are dependent on the signal set $\mathcal{S}$ used. The non-removable singular fade subspaces are {\footnotesize$\left\langle\begin{bmatrix} 1 \\ 0 \end{bmatrix}\right\rangle$} and {\footnotesize$\left\langle\begin{bmatrix}0 \\ 1 \end{bmatrix}\right\rangle,$} which are independent of the signal set used. Owing to the presence of non-removable singular fade subspaces, the overall diversity order of the DNF protocol cannot exceed one. 

From the discussion above, it is clear that there are two classes of singular fade subspaces: removable and non-removable. The non-removable singular fade spaces are created by the channel and is independent of the signal set used. Whatever may be the choice of the network code, the harmful effects of these non-removable singular fade subspaces cannot be mitigated. The harmful effect of the removable singular fade subspaces, which are created by the signal set, can be removed by a proper choice of the adaptive network coding map at R, as in \cite{APT1} and \cite{NVR}--\cite{VNR}. 

To sum up, in the DNF protocol, the transmissions from the nodes A and B are allowed to interfere at R and the effect of MAI is effectively mitigated by adaptively changing the network coding map, thereby removing the harmful effect of all the removable singular fade subspaces.

\subsection{Singular Fade Subspaces for the DSTC Scheme}
\label{sing_DSTC}
%

Let $\mathbf{\Delta {x_A}}=\mathbf{x_A}-\mathbf{x'_A}$ and  $\mathbf{\Delta {x_B}}=\mathbf{x_B}-\mathbf{x'_B} \in \Delta \mathcal{S}^2.$ Then $\mathbf{ C(\mathbf{\Delta} \mathbf{x_A},\mathbf{\Delta} \mathbf{x_B})}=\mathbf{C(\mathbf{x_A},\mathbf{x_B})}-\mathbf{C(\mathbf{x'_A},\mathbf{x'_B})}$ denotes a codeword difference matrix of the DSTC, where $\mathbf{C(\mathbf{x_A},\mathbf{x_B})}$ is the codeword matrix of the DSTC defined in \eqref{DSTC_design}. 
From the discussion in Section \ref{sing_MIMO}, it follows that the singular fade spaces for the proposed DSTC scheme are of the form  $Col^{\perp}\left(\mathbf{ C\left(\mathbf{\Delta} \mathbf{x_A},\mathbf{\Delta} \mathbf{x_B}\right)}\right).$

  


Consider the singular fade subspaces of the form $Col^{\perp}\left(\mathbf{ C\left(0_2,\mathbf{\Delta} \mathbf{x_B}\right)}\right)$ and $Col^{\perp}\left(\mathbf{ C\left(\mathbf{\Delta} \mathbf{x_A},0_2\right)}\right),$ where $\mathbf{\Delta x_A, \Delta x_B \neq 0_2}.$ The first row of the matrix $\mathbf{ C\left(0_2,\mathbf{\Delta} \mathbf{x_B}\right)}$ has both the entries to be zero. Hence, $Col\left(\mathbf{ C\left(0_2,\mathbf{\Delta} \mathbf{x_B}\right)}\right)={\footnotesize \left\langle \begin{bmatrix} 0 \\ 1 \end{bmatrix}\right\rangle}$ and the singular fade subspace $Col^{\perp}\left(\mathbf{ C\left(0_2,\mathbf{\Delta} \mathbf{x_B}\right)}\right)={\footnotesize\left\langle \begin{bmatrix} 1 \\ 0 \end{bmatrix}\right\rangle}.$ By a similar reasoning, $Col^{\perp}\left(\mathbf{C\left(\mathbf{\Delta} \mathbf{x_A},0_2\right)}\right)={\footnotesize\left\langle \begin{bmatrix} 0 \\ 1 \end{bmatrix}\right\rangle.}$ 

If the DSTC codeword matrices are such that  $\text{rank}(\mathbf{ C(\Delta x_A, \Delta x_B)})=2,$ $\forall \mathbf{\Delta x_A} \neq \mathbf{0_{2}}$ and $\mathbf{\Delta x_B} \neq \mathbf{0_{2}},$ all the singular fade subspaces $Col^{\perp}\left(\mathbf{ C\left(\mathbf{\Delta} \mathbf{x_A},\mathbf{\Delta} \mathbf{x_B}\right)}\right)$ collapse to be the trivial singular fade subspace $\left \langle\mathbf{0_2}\right \rangle.$ Equivalently, all the pair-wise error events $\mathbf{C(x_A,x_B)} \rightarrow \mathbf{C(x'_A,x'_B)}, \mathbf{x_A} \neq \mathbf{x'_A},\mathbf{x_B} \neq \mathbf{x'_B} ,$ have diversity order 2. Hence, for a properly chosen DSTC, other than the trivial singular fade subspace, the singular fade subspaces are only the two non-removable singular fade subspaces, while for the DNF protocol, in addition, we had the removable singular fade subspaces.  In this way, by a proper choice of DSTC, the occurrence of the removable singular fade subspaces is avoided at the transmitting nodes itself, without any CSIT.

Hence, we have the following design criterion referred as the \textit{singularity minimization criterion} for DSTCs for two-way relaying: 
\textit{The DSTC codeword difference matrices $\mathbf{C(\Delta x_A, \Delta x_B)}$ need to be full rank for all $\mathbf{\Delta x_A} \neq \mathbf{0_2}$ and $\mathbf{\Delta x_B} \neq \mathbf{0_2},$} to minimize the number of singular fade subspaces. DSTCs satisfying the above criterion are referred as the \textit{singularity minimal} DSTCs.

 Hence for a DSTC which is singularity minimal, the only error events which result in diversity order 1 are of the form $\mathbf{C(x_A,x_B)} \rightarrow \mathbf{C(x_A,x'_B)}, \mathbf{x'_B} \neq \mathbf{x_B}$ and $\mathbf{C(x_A,x_B)} \rightarrow \mathbf{C(x'_A,x_B)}, \mathbf{x'_A} \neq \mathbf{x_A}.$ Hence, the overall coding gain is equal to minimum among all the non-zero singular values of the codeword difference matrices which are of the form $\mathbf{C(0_{2},\Delta x_B)}$ and $\mathbf{C(\Delta x_A,0_{2})}$ \cite{TaSeCa}. Note that the matrices $\mathbf{C(0_{2},\Delta x_B)}$ and $\mathbf{C(\Delta x_A,0_{2})}$ are of rank 1 and have only one non-zero singular value. We have the following coding gain criteria for singularity minimal DSTCs:
\textit{the minimum among all the non-zero singular values of the codeword difference matrices which are of the form $\mathbf{C(0_{2},\Delta x_B)}$ and $\mathbf{C(\Delta x_A,0_{2})}$ needs to be maximized}. 
\begin{example}
Consider the DSTC $\begin{bmatrix}x_{A_1} & 0 \\0 & x_{B_1}\end{bmatrix}.$ This DSTC is nothing but the scheme where A and B transmit in separate time slots, making sure that their transmissions do not interfere at the relay. Even though this DSTC avoids all the removable singular fade subspaces, the end-to-end rate in complex symbols per channel use is less than that of the DNF protocol.
\end{example} 
\subsection{A Construction of Singularity Minimal DSTCs for Algebraic Signal Sets}
A signal set is said to be algebraic if all the signal points of the signal set are algebraic numbers over $\mathbb{Q}$ \footnote{A number is said to be algebraic over $\mathbb{Q}$ if there exists a polynomial with coefficients from $\mathbb{Q}$ of which the number is a root. If there does not exist a polynomial with coefficients from $\mathbb{Q}$ of which the number is a root, the number is said to be transcendental \cite{NJ}.}. All the commonly used signal sets like QAM and PSK are algebraic signal sets. In this subsection, a class of DSTCs which are singularity minimal for algebraic signal sets is provided. 
Let {\footnotesize$\begin{bmatrix} a & b \\ c &d \end{bmatrix}$} be a full rank complex matrix. Consider the class of DSTCs whose codeword matrices are of the form  $\mathbf{C( x_A, x_B)}={\footnotesize\begin{bmatrix} a (x_{A_1}+e^j x_{A_2}) & b (x_{A_1}+e^j x_{A_2}) \\ c (x_{B_1}+e^j x_{B_2}) & d (x_{B_1}+e^j x_{B_2}) \end{bmatrix}}.$ 

\begin{proposition}
The class of DSTCs whose codeword design matrices are of the form given above are singularity minimal for all algebraic signal sets.
\begin{proof}
The proof is as follows: For $\mathbf{\Delta x_A} \neq \mathbf{0_2}$ and $\mathbf{\Delta x_B} \neq \mathbf{0_2},$ at least one of the two components of $\mathbf{\Delta x_A}$ as well as $\mathbf{\Delta x_B}$ should be non-zero. Hence, $(\Delta x_{A_1}+e^j \Delta x_{A_2}) \neq 0$ and $(\Delta x_{B_1}+e^j \Delta x_{B_2}) \neq 0,$ since $e^j$ is transcendental \footnote{By Lindemann-Weierstrass theorem \cite{NJ}, $e^{j q}$ is transcendental for all $q \in \mathbb{Q}$.} whereas $\Delta x_{A_1},\Delta x_{A_2},\Delta x_{B_1}$ and $\Delta x_{B_2}$ are algebraic over $\mathbb{Q}.$  The codeword difference matrix $\mathbf{C( \Delta x_A, \Delta x_B)}$ is full rank for all $\mathbf{\Delta x_A} \neq 0$ and $\mathbf{\Delta x_B} \neq 0,$ since $\text{det}(\mathbf{C( \Delta x_A, \Delta x_B)})=(ad-bc)(\Delta x_{A_1}+e^j \Delta x_{A_2})(\Delta x_{B_1}+e^j \Delta x_{B_2}) \neq 0.$
\end{proof}
\end{proposition}
\begin{example}
\label{DSTC_1}
Consider the case when {\footnotesize$\begin{bmatrix}a &b \\ c & d \end{bmatrix}=\begin{bmatrix}1 &0 \\ 0 & 1 \end{bmatrix}.$} Let 4-PSK be the signal set used at A and B. The DSTC codeword matrix for this case is given by, {\footnotesize$\mathbf{C( x_A, x_B)}=\begin{bmatrix} (x_{A_1}+e^j x_{A_2}) & 0 \\ 0 & (x_{B_1}+e^j x_{B_2})\end{bmatrix}.$} A and B are made to transmit in two different time slots which results in low decoding complexity at R, since A's and B's transmissions can be decoded independently. It can be verified that the coding gain for this DSTC is approximately 0.6877.
\end{example}
\begin{example}
\label{DSTC_2}
Consider the case when {\footnotesize$\begin{bmatrix}a &b \\ c & d \end{bmatrix}=\begin{bmatrix}{1} &{1} \\ -1 &{1} \end{bmatrix}.$} Let 4-PSK be the signal set used at A and B. The DSTC codeword matrix for this case is given by, $\mathbf{C( x_A, x_B)}={\footnotesize \frac{1}{\sqrt{2}}\begin{bmatrix} (x_{A_1}+e^j x_{A_2}) & (x_{A_1}+e^j x_{A_2}) \\ -(x_{B_1}+e^j x_{B_2}) & (x_{B_1}+e^j x_{B_2})\end{bmatrix}}.$ The scaling factor of $\frac{1}{\sqrt{2}}$ is to ensure unit average energy per symbol per time slot. It can be verified that the coding gain for this DSTC is approximately 0.6877, same as that of the DSTC given in Example \ref{DSTC_1}.
\end{example}

The coding gain of the DSTCs given in Examples \ref{DSTC_1} and \ref{DSTC_2} is approximately 0.6877, which is less than the minimum distance of the unit energy 4-PSK signal set, which is $\sqrt{2}.$ In the next section, it is shown that for DSTCs over square QAM and $2^{\lambda}$-PSK signal sets, the coding gain is upper bounded by the minimum distance of the signal set and explicit DSTC constructions which achieve this bound with equality are provided.
\section{Singularity Minimal, Coding Gain Maximal DSTCs over QAM and PSK signal sets}
In this section, it is shown that the coding gain of the DSTCs over square QAM and $2^{\lambda}$-PSK signal sets are upper-bounded by the minimum distance of the signal set. In Subsection  \ref{DSTC_unitary}, a condition under which a singularity minimal DSTC over square QAM and $2^{\lambda}$-PSK signal set meets the upper bound with equality is obtained and explicit constructions of DSTCs are provided. In Subsection \ref{DSTC_unitary_dc}, the constructed DSTC's are shown to be fast ML decodable, i.e., the ML decoding complexity of the constructed DSTCs is shown to be less than the brute-force decoding complexity which is $O(M^4).$

Note that the generator matrices $\mathbf{M_A}$ and $\mathbf{M_B}$ at A and B should be such that the average energy per time slot is unity, i.e., $\mathbb{E}(\parallel \mathbf{x_A} \mathbf{M_A} \parallel ^2) \leq 2$ and $\mathbb{E}(\parallel \mathbf{x_B} \mathbf{M_B} \parallel ^2) \leq 2.$

\begin{lemma}
For singularity minimal DSTC over $\mathcal{S},$ where $\mathcal{S}$ is a square QAM or $2^{\lambda}$-PSK signal set, the coding gain is upper bounded by the minimum distance of the signal set $\mathcal{S}.$ 
\begin{proof}
See Appendix A.
\end{proof}
\end{lemma} 
In the following subsection, the condition under which the upper-bound given in the previous lemma is satisfied with equality is identified and explicit construction of DSTCs are provided.
\subsection{Constructions of Singularity Minimal, Coding Gain Maximal DSTCs over QAM and PSK signal sets}
\label{DSTC_unitary}
The following proposition states that for DSTCs over $\mathcal{S},$ choosing $\mathbf{M_A}$ and $\mathbf{M_B}$ to be unitary matrices ensures that the upper-bound on the coding gain is satisfied with equality, for QAM and PSK signal sets.
\begin{proposition}
For singularity minimal DSTCs over square QAM or $2^{\lambda}$-PSK signal sets, the coding gain is  maximized when the generator matrices $\mathbf{M_A}$ and $\mathbf{M_B}$ at A and B  are unitary matrices.
\end{proposition}
\begin{proof}
When $\mathbf{M_A}$ and $\mathbf{M_B}$ are unitary matrices, $\parallel \mathbf{\Delta x_A M_A} \parallel=\parallel \mathbf{\Delta x_A} \parallel$ and also $\parallel \mathbf{\Delta x_B M_B} \parallel=$ $\parallel \mathbf{\Delta x_B} \parallel.$  Hence, $\displaystyle{\min_{\substack {{\mathbf{\Delta x_A} \in \Delta \mathcal{S}^2,}\\{ \mathbf{\Delta x_A} \neq 0_2}}} \parallel \mathbf{\Delta x_A}\mathbf{M_A} \parallel= \min_{\substack{{\Delta x_{A_1}} \in \Delta \mathcal{S}},\\{\Delta x_{A_1} \neq 0}}\vert \Delta x_{A_1} \vert=d_{min}(\mathcal{S})}$ and similarly \\$\displaystyle{\min_{\substack {{\mathbf{\Delta x_B} \in \Delta \mathcal{S}^2,}\\{ \mathbf{\Delta x_B} \neq 0_2}}} \parallel \mathbf{\Delta x_B}\mathbf{M_B} \parallel=d_{min}(\mathcal{S})},$ where $d_{min}(\mathcal{S})$ denotes the minimum distance of $\mathcal{S}.$
 
  The coding gain of the DSTC is the minimum among all the non-zero singular values of the codeword difference matrices which are of the form $\mathbf{C(0_{2},\Delta x_B)}$ and $\mathbf{C(\Delta x_A,0_{2})},$ i.e., the coding gain is equal to 
 {\footnotesize $\displaystyle{\min \left\lbrace \min_{\substack {{\mathbf{\Delta x_A} \in \Delta \mathcal{S}^2,}\\{ \mathbf{\Delta x_A} \neq 0_2}}} \parallel \mathbf{\Delta x_A}\mathbf{M_A} \parallel, \min_{\substack {{\mathbf{\Delta x_B} \in \Delta \mathcal{S}^2,}\\{ \mathbf{\Delta x_B} \neq 0_2}}} \parallel \mathbf{\Delta x_B}\mathbf{M_B} \parallel \right \rbrace,}$} which is equal to $d_{min}(\mathcal{S}).$
\end{proof}
In the following examples, constructions of singularity minimal DSTCs whose generator matrices are unitary are provided.
\begin{construction}
\label{DSTC_ldc1}
Consider the DSTC over $\mathcal{S}$ for which {\footnotesize$\mathbf{M_A}=\frac{1}{\sqrt{5}}\begin{bmatrix}\alpha & \bar{\alpha} \\ {\alpha\phi} & \bar{\alpha} \bar{\phi}\end{bmatrix}$} and {\footnotesize$\mathbf{M_B}=\frac{1}{\sqrt{5}}\begin{bmatrix}j\alpha & \bar{\alpha} \\ j{\alpha\phi} & \bar{\alpha} \bar{\phi}\end{bmatrix},$} where $\phi=\frac{1+\sqrt{5}}{2},$ $\bar{\phi}=\frac{1-\sqrt{5}}{2},$ $\alpha=1+j-j\phi$ and $\bar{\alpha}=1+j-j \bar{\phi}.$ The DSTC codeword matrix is of the form $\mathbf{C(x_A,x_B)}={\footnotesize\begin{bmatrix}\mathbf{x_A}{\mathbf{M_A}} \\\mathbf{x_B}{\mathbf{M_B}} \end{bmatrix}}.$ The codeword difference matrix $\mathbf{C(\Delta x_A, \Delta x_B)}$ is full rank for all $\mathbf{\Delta x_A}\neq 0$ and $\mathbf{\Delta x_B}\neq 0,$ when the signal points belong to $\mathbb{Z}[j]$ \cite{MaBe}. Hence the DSTC is singularity minimal for all signal sets whose signal points belong to $\mathbb{Z}[j].$ Also, since $\mathbf{M_A}$ and $\mathbf{M_B}$ are unitary, for square QAM signal set, the DSTC maximizes the coding gain.
\end{construction}
\begin{note}
The DSTC given in Construction \ref{DSTC_ldc1} was constructed in \cite{MaBe} towards satisfying the design criterion formulated in \cite{GaBo} for the two-user non-cooperative Multiple Access Channel (MAC). In \cite{MaBe}, the DSTC given in the above example was shown to be DMT optimal for two-user MAC.
\end{note}
\begin{construction}
\label{DSTC_ldc2}
Consider the DSTC for which {\footnotesize$\mathbf{M_A}=\mathbf{I_{2}}$} and {$\mathbf{M_B}=\begin{bmatrix}\cos\phi_g & -\sin\phi_g e^{j\theta}\\ \sin\phi_g & \cos\phi_g e^{j \theta} \end{bmatrix},$} where $\phi_g=\tan^{-1} \sqrt{5}.$ The DSTC codeword matrix $\mathbf{C( x_A,  x_B)}$ is given by, {\vspace{-.2 cm} \footnotesize $$\begin{bmatrix} x_{A_1} & x_{A_2} \\ x_{B_1} \cos\phi_g+ x_{B_2}\sin\phi_g & e^{j \theta} (-x_{B_1}\sin\phi_g+ x_{B_2}\cos\phi_g) \end{bmatrix}.$$} For a complex number $a,$ let $\mathbb{Q}(a)$ denote the smallest field containing $\mathbb{Q}$ and $a.$ It is shown in Lemma \ref{const2_proof} below that choosing $\theta=\frac{\pi}{4}$ ensures singularity minimality for signal sets (for example QAM) whose signal points belong to $\mathbb{Q}(j)$ and choosing $\theta=\frac{\pi}{2^{\lambda}}$ ensures singularity minimality for signal sets (for example $2^{\lambda}$-PSK) whose signal points belong to $\mathbb{Q}(e^{j \frac{2 \pi}{2^{\lambda}}}).$ Also, since $\mathbf{M_A}$ and $\mathbf{M_B}$ are unitary, this DSTC maximizes the coding gain, for square QAM and $2^{\lambda}$-PSK signal sets. The advantage of this construction over Construction \ref{DSTC_ldc1} is that encoding at node A is simple, since it does not involve any linear combination of $x_{A_1}$ and $x_{A_2}.$
\end{construction}
\begin{lemma}
\label{const2_proof}
For the DSTC given in construction \ref{DSTC_ldc2}, choosing $\theta=\frac{\pi}{4}$ ensures singularity minimality for signal sets whose points belong to $\mathbb{Q}(j)$ and choosing $\theta=\frac{\pi}{2^{\lambda}}$ ensures singularity minimality for signal sets whose signal points belong to $\mathbb{Q}(e^{j \frac{2 \pi}{2^{\lambda}}}).$
\begin{proof}
The proof is given for the case when the signal points belong to $\mathbb{Q}(j).$ The proof for the case when the signal points belong to $\mathbb{Q}(e^{j \frac{2 \pi}{2^{\lambda}}})$ is exactly similar and is omitted.\\
Let $\Delta x_{A_i}=x_{A_i}-x'_{A_i}$ and $\Delta x_{B_i}=x_{B_i}-x'_{B_i},$ where $x_{A_i},x'_{A_i},x_{B_i},x'_{B_i} \in \mathcal{S} \subset \mathbb{Q}(j).$ and $i \in \lbrace 1,2 \rbrace.$ To prove singularity minimality, it needs to be shown that when at least one out of $\Delta x_{A_1}$ and $\Delta x_{A_2}$ ($\Delta x_{B_1}$ and $\Delta x_{B_2}$) is non-zero, the codeword difference matrix is full rank. The ratios $\frac{\Delta x_{B_1}}{\Delta x_{B_2}}$ and $-\frac{\Delta x_{B_2}}{\Delta x_{B_1}}$ belong to $\mathbb{Q}(j)$ while $\tan \phi_g= \sqrt{5}$ does not belong to $\mathbb{Q}(j).$ Hence, $\Delta x_{B_1} \cos\phi_g+ \Delta x_{B_2}\sin\phi_g \neq 0$ and $-\Delta x_{B_1}\sin\phi_g+ \Delta x_{B_2}\cos\phi_g \neq 0.$ Since $\sin \phi_g=\frac{\sqrt{5}}{\sqrt{6}}$ and $\cos \phi_g=\frac{1}{\sqrt{6}},$ $\Delta x_{A_2} (\Delta x_{B_1} \cos\phi_g+ \Delta x_{B_2}\sin\phi_g)$ and $\Delta x_{A_1}(-\Delta x_{B_1} \sin\phi_g+ \Delta x_{B_2}\cos\phi_g)$ belong to $\mathbb{Q}(j,\sqrt{5},\sqrt{6}),$ where $\mathbb{Q}(j,\sqrt{5},\sqrt{6})$ denotes the smallest filed containing $\mathbb{Q},j,\sqrt{5}$ and $\sqrt{6}.$ The determinant of the codeword difference matrix is given by, $$\Delta x_{A_1}e^{j\frac{\pi}{4}}(-\Delta x_{B_1} \sin\phi_g+ \Delta x_{B_2}\cos\phi_g)-\Delta x_{A_2} (\Delta x_{B_1} \cos\phi_g+ \Delta x_{B_2}\sin\phi_g).$$ The determinant is non-zero since the ratio $\frac{\Delta x_{A_2} (\Delta x_{B_1} \cos\phi_g+ \Delta x_{B_2}\sin\phi_g)}{\Delta x_{A_1}(-\Delta x_{B_1} \sin\phi_g+ \Delta x_{B_2}\cos\phi_g)}$ belongs to $\mathbb{Q}(j,\sqrt{5},\sqrt{6}),$ while $e^{j \frac{\pi}{4}}$ does not belong to $\mathbb{Q}(j,\sqrt{5},\sqrt{6}).$
\end{proof}
\end{lemma}
\subsection{Decoding Complexity of Singularity Minimal, Maximal Coding Gain DSTCs  over $\mathcal{S}$}
\label{DSTC_unitary_dc}
After the two MA phases, R jointly decodes for the two message vectors $\mathbf{x_A}$ and $\mathbf{x_B}$ of A and B respectively. In general, the complexity of this joint ML decoding at R is $O(M^4),$ where $M$ is the cardinality of the signal set $\mathcal{S}.$ The choice of the generator matrices $\mathbf{M_A}$ and $\mathbf{M_B}$ being unitary not only maximizes the coding gain for QAM and PSK signal sets, but also results in a reduced decoding complexity at R.
 
 The following proposition states that when conditional ML decoding \cite{BiHoVi}, \cite{SrRa} is employed, the decoding complexity of the DSTCs constructed in the previous section for which the generator matrices $\mathbf{M_A}$ and $\mathbf{M_B}$ are unitary is $O(M^3)$ for any arbitrary signal set and is $O(M^2)$ for square QAM signal set. Note that the brute force decoding complexity is $O(M^4).$ 

\begin{proposition}
When the generator matrices of the singularity minimal DSTC over $\mathcal{S}$ are unitary, the decoding complexity using conditional ML decoding is $O(M^3)$ when the signal set $\mathcal{S}$ is arbitrary and is $O(M^2)$ when the signal set $\mathcal{S}$ is square QAM.
\begin{proof}
See Appendix B.
\end{proof}
\end{proposition} 

Compared with the DNF protocol, the decoding complexity is more for singularity minimal coding gain maximal DSTCs  over $\mathcal{S}.$ For the DNF protocol, the decoding complexity is $O(M^2)$ for non square QAM signal sets while it is $O(M)$ for square QAM signal set
\footnote{
For the DNF protocol, with QAM signal set, conditioning on $x_A,$ $x_B$ can be decoded with constant decoding complexity by rounding off to the nearest integer, which results in an overall decoding complexity of $O(M).$}. As indicated by the simulation results in the next section, the proposed DSTC offers slightly better performance than the adaptive network coding scheme and eliminates the need for adaptive switching of network coding maps at R. But this comes at the cost of increased decoding complexity at R. 

\section{Simulation Results}
  All the simulation results presented are for the case when the end nodes use 4-PSK signal set. By `DSTC 1' and `DSTC 2' we refer to the DSTCs given in Construction \ref{DSTC_ldc1} and Construction \ref{DSTC_ldc2} respectively. As a reference scheme, we consider the scheme in which XOR network code is used irrespective of channel conditions and no DSTC is employed, which is referred as `XOR N/W code'. Assuming unit noise variances at all the nodes, the average energies of the transmissions at the nodes, which are assumed to be equal, is defined to be the Signal to Noise Ratio (SNR).  The proposed DSTC scheme is also compared with the adaptive network coding schemes proposed in \cite{APT1} and \cite{NVR}-\cite{VNR}. Since for 4-PSK signal set, the adaptive network coding scheme based on the Nearest Neighbour Clustering (NNC) algorithm proposed in \cite{APT1} and the scheme based on Latin Squares proposed in \cite{NVR}-\cite{VNR} turn out to be the same, without distinguishing them we refer to both as `adaptive N/W code'. Fig. \ref{fig:ber_comp1} shows the SNR vs BER performance for different schemes for the case when all the fading coefficients are i.i.d. and Rayleigh distributed. In Fig. \ref{fig:ber_comp2} and Fig. \ref{fig:ber_comp3} similar plots are shown for a Rician fading scenario with Rician factors \footnote{Rician factor is the power ratio between the line of sight and scattered components.} of 0 dB and 5 dB respectively.  From Fig. \ref{fig:ber_comp1}-\ref{fig:ber_comp3}, it can be seen that the diversity order is one for all the schemes. Also, it can be seen that at high SNR, both `DSTC 1' as well as `DSTC 2' offer nearly the same performance and they perform better than the `XOR N/W code' as well as the `adaptive N/W code'. For a Rayleigh fading scenario, at high SNR, the DSTCs offer a gain of 2 dB over `XOR N/W code' while the `adaptive N/W code' offers a gain of about 0.5 dB over the `XOR N/W code'. For a Rician factor of 0 dB, at high SNR, the DSTCs offer a gain of 2 dB over `XOR N/W code' while the `adaptive N/W code' offers a gain of about 1.2 dB over the `XOR N/W code'. For a Rician factor of 5 dB, at high SNR, the DSTCs offer a gain of 5.5 dB over 'XOR N/W code' while the `adaptive N/W code' offers a gain of about 4 dB over the `XOR N/W code'. The reason why the DSTC based scheme performs better than the adaptive N/W coding scheme is as follows: during the BC phase always a 4 point signal set is used for the DSTC based scheme, while depending on channel conditions 4 point or 5 point signal set is used for the adaptive network coding scheme \cite{APT1},\cite{NVR}.  

\section{Discussion}
A DSTC scheme was proposed for the two-way relaying scenario. It was shown that deep channel fades occur when the channel fade coefficient vector falls in a finite number of vector subspaces called the singular fade subspaces. The connection between the dimension of these vector subspaces and the transmit diversity order was established. Design criterion to minimize the number of singular fade subspaces for the DSTC scheme and maximize the coding gain were obtained. Explicit low decoding complexity constructions of DSTCs were provided. The problem of constructing singularity minimal DSTCs with decoding complexity same as that of the DNF protocol, without sacrificing the coding gain, remains open. Extending the DSTC scheme for two-way relaying with multiple antennas and multi-way relaying are possible directions for future work.
\section*{Acknowledgement}
This work was supported  partly by the DRDO-IISc program on Advanced Research in Mathematical Engineering through a research grant as well as the INAE Chair Professorship grant to B.~S.~Rajan.


\begin{appendices}
\section{Proof of Lemma 1}
Since $\mathbf{M_A} \mathbf{M^H_A}$ is Hermitian, it is unitarily diagonalizable, i.e., {\small $\mathbf{M_A} \mathbf{M^H_A}= \mathbf{U_A} \mathbf{\Lambda_A} \mathbf{U^H_A},$} where $\mathbf{U_A}$ is a unitary matrix and  $\mathbf{\Lambda_A}$ is a diagonal matrix with diagonal entries denoted as $\lambda_{A_1}$ and $\lambda_{A_2}.$ Note that $\lambda_{A_1}$ and $\lambda_{A_2}$ are non-negative since $\mathbf{M_A} \mathbf{M^H_A}$ is positive semi-definite. 
Let $\mathbf{M_A}=\begin{bmatrix} a_{11} & a_{12} \\ a_{21} & a_{22}\end{bmatrix}.$ We have, 
{\footnotesize
$\mathbb{E}(\parallel \mathbf{x_A} \mathbf{M_A} \parallel ^2) =\vert a_{11} \vert ^2 \mathbb{E}(\vert x_{A_1} \vert ^2)+\vert a_{12} \vert ^2 \mathbb{E}(\vert x_{A_2} \vert ^2)+\vert a_{21} \vert ^2 \mathbb{E}(\vert x_{A_1} \vert ^2)+\vert a_{22} \vert ^2 \mathbb{E}(\vert x_{A_2} \vert ^2)=\vert  a_{11} \vert ^2+\vert  a_{12} \vert ^2+\vert  a_{21} \vert ^2+\vert  a_{22} \vert ^2,$
} since $\mathbb{E}(x_{A_1} x^*_{A_2})=\mathbb{E}(x_{A_2} x^*_{A_1})=0$ for square QAM and $2^{\lambda}$-PSK signal sets. Since $\mathbb{E}(\parallel \mathbf{x_A} \mathbf{M_A} \parallel ^2)\leq 2,$ we have  $\vert  a_{11} \vert ^2+\vert  a_{12} \vert ^2+\vert  a_{21} \vert ^2+\vert  a_{22} \vert ^2=Trace(\mathbf{M_A} \mathbf{M_A}^H)=\lambda_{A_1}+\lambda_{A_2}\leq 2.$

 The coding gain of the DSTC is the minimum among all the non-zero singular values of the codeword difference matrices which are of the form $\mathbf{C(0_{2},\Delta x_B)}$ and $\mathbf{C(\Delta x_A,0_{2})},$ i.e., the coding gain is equal to 
 {\footnotesize $\displaystyle{\min \left\lbrace \min_{\substack {{\mathbf{\Delta x_A} \in \Delta \mathcal{S}^2,}\\{ \mathbf{\Delta x_A} \neq 0_2}}} \parallel \mathbf{\Delta x_A}\mathbf{M_A} \parallel, \min_{\substack {{\mathbf{\Delta x_B} \in \Delta \mathcal{S}^2,}\\{ \mathbf{\Delta x_B} \neq 0_2}}} \parallel \mathbf{\Delta x_B}\mathbf{M_B} \parallel \right \rbrace.}$}
 
Let $d_{min}(\mathcal{S})$ denote the minimum distance of the signal set $\mathcal{S}.$ 
 
 Consider $\parallel \mathbf{\Delta x_A}\mathbf{M_A} \parallel^2= \Delta \mathbf{x_A} \mathbf{M_A} \mathbf{M_A}^H \Delta\mathbf{x_A}^H= \mathbf{\Delta \tilde{x}_A} \mathbf{\Lambda_A} \mathbf{\Delta \tilde{x}_A}^H=\lambda_{A_1} \vert \Delta \tilde{x}_{A_1} \vert^2 +\lambda_{A_2} \vert \Delta \tilde{x}_{A_2} \vert^2,$ where $\mathbf{\Delta \tilde{x}_A} =\mathbf{\Delta {x}_A} \mathbf{U_A}\triangleq[\Delta \tilde{x}_{A_1} \; \Delta \tilde{x}_{A_2}].$ 
 
 Let $\mathbf{u_{A_1}}=[u_{A_{11}} \; u_{A_{12}}]$ and  $\mathbf{u_{A_2}}=[u_{A_{21}} \; u_{A_{22}}]$ denote the rows of $\mathbf{U_A}.$ For $\mathbf{\Delta x_A}=[\Delta x_{A_1} \; 0],$ $\parallel \mathbf{\Delta {x}_A}\mathbf{M_A} \parallel^2= \vert \Delta x_{A_1} \vert^2(\vert u_{A_{11}}\vert ^2 \lambda_{A_1}+\vert u_{A_{12}}\vert ^2 \lambda_{A_2}).$ 
 
 Hence, we have, {\footnotesize $\displaystyle{\min_{\substack {{\mathbf{\Delta x_A} \in \Delta \mathcal{S}^2,}\\{ \mathbf{\Delta x_A} \neq 0_2}}} \parallel \mathbf{\Delta x_A}\mathbf{M_A} \parallel^2 \leq d^2_{min}(\mathcal{S}) (\vert u_{A_{11}}\vert ^2 \lambda_{A_1}+\vert u_{A_{12}}\vert ^2 \lambda_{A_2})}.$} 
  Similarly, we have, \\{\footnotesize $\displaystyle{\min_{\substack {{\mathbf{\Delta x_A} \in \Delta \mathcal{S}^2,}\\{ \mathbf{\Delta x_A} \neq 0_2}}} \parallel \mathbf{\Delta x_A}\mathbf{M_A} \parallel^2 \leq d^2_{min}(\mathcal{S})(\vert u_{A_{21}}\vert ^2 \lambda_{A_1}+\vert u_{A_{22}}\vert ^2 \lambda_{A_2}).}$}
 Since $\mathbf{U_A}$ is unitary $\vert u_{A_{11}} \vert ^2 = \vert u_{A_{22}} \vert ^2$ and $\vert u_{A_{12}} \vert ^2 = \vert u_{A_{21}} \vert ^2.$ Therefore, we have,
 
{\vspace{-.4 cm}
\footnotesize
\begin{align} 
\label{dmin_ub}
 \min_{\substack {{\mathbf{\Delta x_A} \in \Delta \mathcal{S}^2,}\\{ \mathbf{\Delta x_A} \neq 0_2}}} \parallel \mathbf{\Delta x_A}\mathbf{M_A} \parallel^2 &\leq d^2_{min}(\mathcal{S}) \min \lbrace(\vert u_{A_{11}}\vert ^2 \lambda_{A_1}+\vert u_{A_{12}}\vert ^2 \lambda_{A_2}),(\vert u_{A_{11}}\vert ^2 \lambda_{A_2}+\vert u_{A_{12}}\vert ^2 \lambda_{A_1}) \rbrace.
 \end{align}
 \vspace{-.4 cm}
}  

Since $\mathbf{U_A}$ is unitary, $\vert u_{A_{11}}\vert ^2=1-\vert u_{A_{12}}\vert ^2.$ For a given $\lambda_{A_1}$ and $\lambda_{A_2},$ the upper-bound in \eqref{dmin_ub} is maximized over all $\vert u_{A_{11}}\vert ^2$ when the two terms inside $\min$ are equal, i.e., $\vert u_{A_{11}}\vert ^2 \lambda_{A_1}+\vert u_{A_{12}}\vert ^2 \lambda_{A_2}=\vert u_{A_{11}}\vert ^2 \lambda_{A_2}+\vert u_{A_{12}}\vert ^2 \lambda_{A_1},$ for which $\vert u_{A_{11}}\vert ^2=\frac{1}{2}$ and this maximum value is equal to $d^2_{min}(\mathcal{S})\frac{(\lambda_{A_1}+\lambda_{A_2})}{2}.$ Since,  $ \lambda_{A_1}+\lambda_{A_2} \leq 2,$ the maximum value of the upper-bound in \eqref{dmin_ub} is less than or equal to $d^2_{min}(\mathcal{S}).$ Hence, $\displaystyle{\min_{\substack {{\mathbf{\Delta x_B} \in \Delta \mathcal{S}^2,}\\{ \mathbf{\Delta x_B} \neq 0_2}}} \parallel \mathbf{\Delta x_A}\mathbf{M_A} \parallel}\leq d_{min}(\mathcal{S}).$ Similarly, it can be shown that $\displaystyle{\min_{\substack {{\mathbf{\Delta x_A} \in \Delta \mathcal{S}^2,}\\{ \mathbf{\Delta x_B} \neq 0_2}}} \parallel \mathbf{\Delta x_B}\mathbf{M_B} \parallel}$ is also upper-bounded by $d_{min}(\mathcal{S}).$ Hence, the coding gain of the DSTC over square QAM or $2^{\lambda}$-PSK signal set is upper-bounded by $d_{min}(\mathcal{S}).$ This completes the proof.
\section{Proof of Proposition 3}
To prove the proposition, we adopt a procedure similar to the one used in \cite{SrRa}. 

Let $\mathbf{\tilde{y}_R}=[y_{R_1}^R \;y_{R_1}^I\;y_{R_2}^R \;y_{R_2}^I]^T,$ $\mathbf{\tilde{x}}=[x_{A_1}^R ~ x_{A_1}^I ~ x_{A_2}^R ~  x_{A_2}^I ~ x_{B_1}^R  ~ x_{B_1}^I ~ x_{B_2}^R ~ x_{B_2}^I]^T$ and $\mathbf{\tilde{z}_R}=[z_{R_1}^R \;z_{R_1}^I\;z_{R_2}^R \;z_{R_2}^I]^T.$ The vector $\mathbf{\tilde{y}_R}$ can be written as $\mathbf{\tilde{y}_R}=\mathbf{H_{eq}}\tilde{\mathbf{x}}+\mathbf{\tilde{z}_R},$ where $\mathbf{H_{eq}}$ is a $4 \times 8$ real matrix whose entries are functions of $h_A$ and $h_B,$ determined by the DSTC. Using $\mathbf{QR}$ decomposition, the matrix $\mathbf{H_{eq}}$ can be decomposed as $\mathbf{H_{eq}}=\mathbf{QR},$ where $\mathbf{Q} \in \mathbb{R}^{4 \times 4}$ is a orthogonal matrix and $\mathbf{R} \in \mathbb{R}^{4 \times 8}$ can be written as $[\mathbf{R_1} \; \mathbf{R_2}],$ with $\mathbf{R_1},\mathbf{R_2} \in \mathbb{R}^{4 \times 4},$ $\mathbf{R_1}$ being an upper-triangular matrix. The joint ML decoding metric at R is given by $\parallel \mathbf{\tilde{y}_R}- \mathbf{H_{eq}}\tilde{\mathbf{x}} \parallel=\parallel \mathbf{Q^T}\mathbf{\tilde{y}_R}- \mathbf{R}\tilde{\mathbf{x}} \parallel=\parallel {\mathbf{y'_R}}- \mathbf{R}\tilde{\mathbf{x}} \parallel,$ where $\mathbf{y'_R}=\mathbf{Q^T}\mathbf{\tilde{y}_R}.$
 
 For a singularity minimal DSTC  over $\mathcal{S},$ let the generator matrices be $\mathbf{M_A}=\mathbf{U_A}$ and $\mathbf{M_B}=\mathbf{U_B},$ where $\mathbf{U_A}$ and $\mathbf{U_B}$ are unitary matrices. Let $\mathbf{u_{A_i}}$ and $\mathbf{u_{B_i}}$ denote the $i^{th}$ rows of $\mathbf{U_A}$ and $\mathbf{U_B}$ respectively. Then the weight matrices of the DSTC defined in \eqref{DSTC_weight} are given by, $\mathbf{W_{A_i}^R}=j\mathbf{W_{A_i}^I}={\footnotesize \begin{bmatrix}\mathbf{u_{A_i}} \\ 0_2^T \end{bmatrix}}$ and 
 $\mathbf{W_{B_i}^R}=j\mathbf{W_{B_i}^I}={\footnotesize\begin{bmatrix}0_2^T \\ \mathbf{u_{B_i}} \end{bmatrix}}.$ 
We have, $\mathbf{W_{A_1}^R}\mathbf{W_{A_1}^I}^\mathbf{H}={\footnotesize\begin{bmatrix}\mathbf{u_{A_1}} \\ 0_2^T \end{bmatrix}\begin{bmatrix}j\mathbf{u_{A_1}^H} & 0_2 \end{bmatrix}=\begin{bmatrix}j & 0 \\ 0 & 0 \end{bmatrix}}$ and similarly, $\mathbf{W_{A_1}^I}\mathbf{W_{A_1}^R}^\mathbf{H}={\footnotesize\begin{bmatrix}-j & 0 \\ 0 & 0 \end{bmatrix}}.$ Hence, $\mathbf{W_{A_1}^R}\mathbf{W_{A_1}^I}^\mathbf{H}+\mathbf{W_{A_1}^I}\mathbf{W_{A_1}^R}^\mathbf{H}=\mathbf{O_2,}$ where $\mathbf{O_2}$ denotes the $2 \times 2$ null matrix. Also, $\mathbf{W_{A_1}^R}\mathbf{{W_{A_2}^R}^H}=\mathbf{O_2},$ since $\mathbf{u_{A_1}}$ and $\mathbf{u_{A_2}}$ are orthogonal vectors. Hence, $\mathbf{W_{A_1}^R}\mathbf{W_{A_2}^R}^\mathbf{H}+\mathbf{W_{A_2}^R}\mathbf{W_{A_1}^R}^\mathbf{H}=\mathbf{O_2}$.  
Similarly, using the fact that $\mathbf{U_A}$ and $\mathbf{U_B}$ are unitary matrices, it can be shown that the following pairs of matrices are also Hurwitz-Radon orthogonal\footnote {Two matrices $\mathbf{M_1}$ and $\mathbf{M_2}$ are said to be Hurwitz-Radon orthogonal if $\mathbf{M_1} \mathbf{M_2^H} + \mathbf{M_2} \mathbf{M_1^H} =0.$}: 
$\lbrace\mathbf{W_{A_1}^R},\mathbf{W_{A_2}^I} \rbrace,$ $\lbrace \mathbf{W_{A_1}^I},$ $ \mathbf{W_{A_2}^R} \rbrace,$ $\lbrace \mathbf{W_{A_1}^I}, \mathbf{W_{A_2}^I} \rbrace,$ $\lbrace \mathbf{W_{A_2}^R}, \mathbf{W_{A_2}^I} \rbrace,$ $\lbrace \mathbf{W_{B_1}^R},$ $\mathbf{W_{B_1}^I} \rbrace,$ $\lbrace \mathbf{W_{B_1}^R},$ $\mathbf{W_{B_2}^R} \rbrace,$ $\lbrace \mathbf{W_{B_1}^R},$ $\mathbf{W_{B_2}^I} \rbrace,$ $\lbrace \mathbf{W_{B_1}^I},$ $ \mathbf{W_{B_2}^R} \rbrace,$ $\lbrace \mathbf{W_{B_1}^I}, \mathbf{W_{B_2}^I} \rbrace,$ $\lbrace\mathbf{W_{B_2}^R}, \mathbf{W_{B_2}^I} \rbrace.$ 
 
Let $\tilde{x}_i$ denote the $i^{th}$ component of the vector $\mathbf{\tilde{x}}.$ The $i^{th}$ and $j^{th}$ columns of $\mathbf{H_{eq}}$ are orthogonal and hence the $(i,j)^{th}$ entry of $\mathbf{R}$ ($i \leq j$) is zero for all realizations of $h_A$ and $h_B,$ if and only if the weight matrices of the DSTC corresponding to the symbols $\tilde{x}_i$ and $\tilde{x}_j$ are Hurwitz-Radon orthogonal (follows from Theorem 2, \cite{SrRa})\footnote {Theorem 2 in \cite{SrRa} proves only the `if' part. However, following an approach similar to the proof given in \cite{SrRa}, it is easy to show that the weight matrices of the DSTC corresponding to the symbols $\tilde{x}_i$ and $\tilde{x}_j$ need to be Hurwitz-Radon orthogonal, for the $(i,j)^{th}$ entry of $\mathbf{R}$ ($i \leq j$) to be zero for all realizations of $h_A$ and $h_B,$ and hence the `only if' part also holds.}. Hence the matrix $\mathbf{R}$ is of the form given below. 
 

{\vspace{-.4 cm}
\footnotesize 
 \begin{align}
 \label{R_matrix1}
 \mathbf{R}=\begin{bmatrix} * & 0 & 0 & 0 & * & *& * &*\\
0& * & 0 & 0 & * & * & *& * \\
0& 0 & * & 0 & * & * & *& * \\
0& 0 & 0 & * & * & * & *& * \\
\end{bmatrix}
\end{align}
\vspace{-.4 cm}}

Note that $*$ denotes possible non-zero entries. The claim is that all the entries denoted by $*$ are non-zeros. It is clear that all the diagonal entries are non-zeros. For the $(1,5)^{th}$ entry in \eqref{R_matrix1} to be a zero, $\mathbf{W_{A_1}^R} \mathbf{W_{B_1}^R}^{\mathbf{H}}+ \mathbf{W_{B_1}^R} \mathbf{W_{A_1}^R}^{\mathbf{H}}=\footnotesize{\begin{bmatrix} 0 & \mathbf{u_{A_1} u_{B_1}^H} \\ \mathbf{u_{B_1} u_{A_1}^H} & 0 \end{bmatrix}}=0,$ which implies that $\mathbf{u_{A_1}}$ and $\mathbf{u_{B_1}}$ are orthogonal vectors. Then the vector $\mathbf{u_{B_1}}$ should belong to the one-dimensional subspace which is orthogonal to $\mathbf{u_{A_1}}.$ Since $\mathbf{u_{A_2}}$ also belongs to this one-dimensional subspace and both $\mathbf{u_{B_1}}$ as well as $\mathbf{u_{A_2}}$ are of unit norm, $\mathbf{u_{B_1}}=e^{j \theta} \mathbf{u_{A_2}},$ for some angle $\theta.$ In that case, the DSTC codeword difference matrix is of the form ${\footnotesize \begin{bmatrix} \Delta x_{A_1} \mathbf{u_{A_1}}+ \Delta x_{A_2} \mathbf{u_{A_2}}\\\Delta x_{B_1} e^{j \theta}\mathbf{u_{A_2}}+ \Delta x_{B_2} \mathbf{u_{B_2}} \end{bmatrix}},$ which is not full rank when $\Delta x_{A_2}, \Delta x_{B_1} \neq 0,$ $\Delta x_{A_1}=\Delta x_{B_2}=0$ and hence the singularity minimization criterion is violated. Hence, $(1,5)^{th}$ entry shown by  $*$ in \eqref{R_matrix1} is non-zero. By a similar argument, it can be shown that the other non-diagonal entries denoted by $*$ in \eqref{R_matrix1} are non-zeros. 

 From  the matrix $\mathbf{R}$ given in \eqref{R_matrix1}, it can be seen that conditioning on the variables $x_{B_1}$ and $x_{B_2},$ the symbols $x_{A_1}$ and $x_{A_2}$ can be decoded independently \cite{SrRa}. Since the total number of choices for $x_{B_1}$ and $x_{B_2}$ is $M^2$ and independently decoding $x_{A_1}$ and $x_{A_2}$ requires $2M$ computations, the decoding involves $2M^3$ computations and hence the decoding complexity at R is $O(M^3).$ 
  
 For square QAM signal sets, the decoding complexity can be further reduced, since the real and imaginary parts independently take values. From \eqref{R_matrix1}, it can be seen that conditioning on $x_{B_1}$ and $x_{B_2},$ the real and imaginary parts of $x_{A_1}$ as well as $x_{A_2}$ can be decoded independently. Since decoding the real and imaginary points of a signal point in QAM signal set is of constant complexity independent of $M$(decoding can be done by rounding off to the nearest integer \cite{SrRa}), the ML decoding complexity is $O(M^2)$ for square QAM signal sets. This completes the proof.
\end{appendices}
\newpage
\begin{figure}[htbp]
\vspace{-.4 cm}
\centering
\subfigure[ MA Phase]{
\includegraphics[totalheight=1in,width=2.5in]{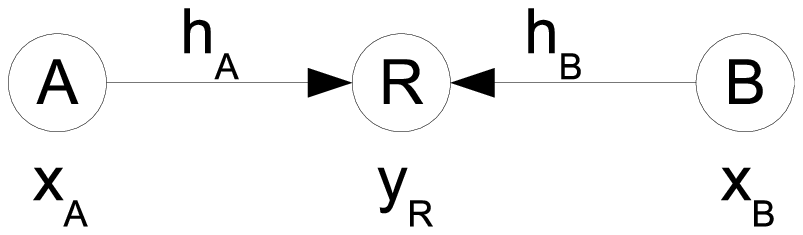}
\label{DNF_MAC}}
\subfigure[ BC Phase]{
\includegraphics[totalheight=1in,width=2.5in]{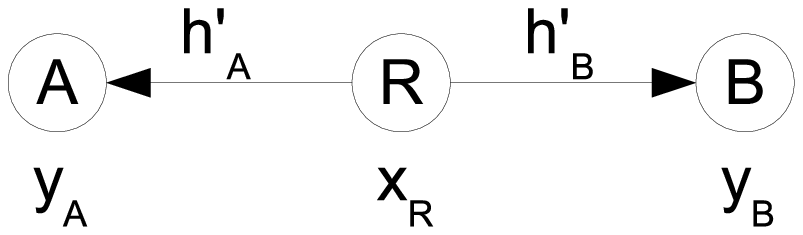}
\label{DNF_BC}}
\caption{Wireless two-way relaying}
\label{DNF_protocol}
\end{figure}
\begin{figure*}[htbp]
\centering
\vspace{-.2 cm}
\includegraphics[totalheight=4in,width=7in]{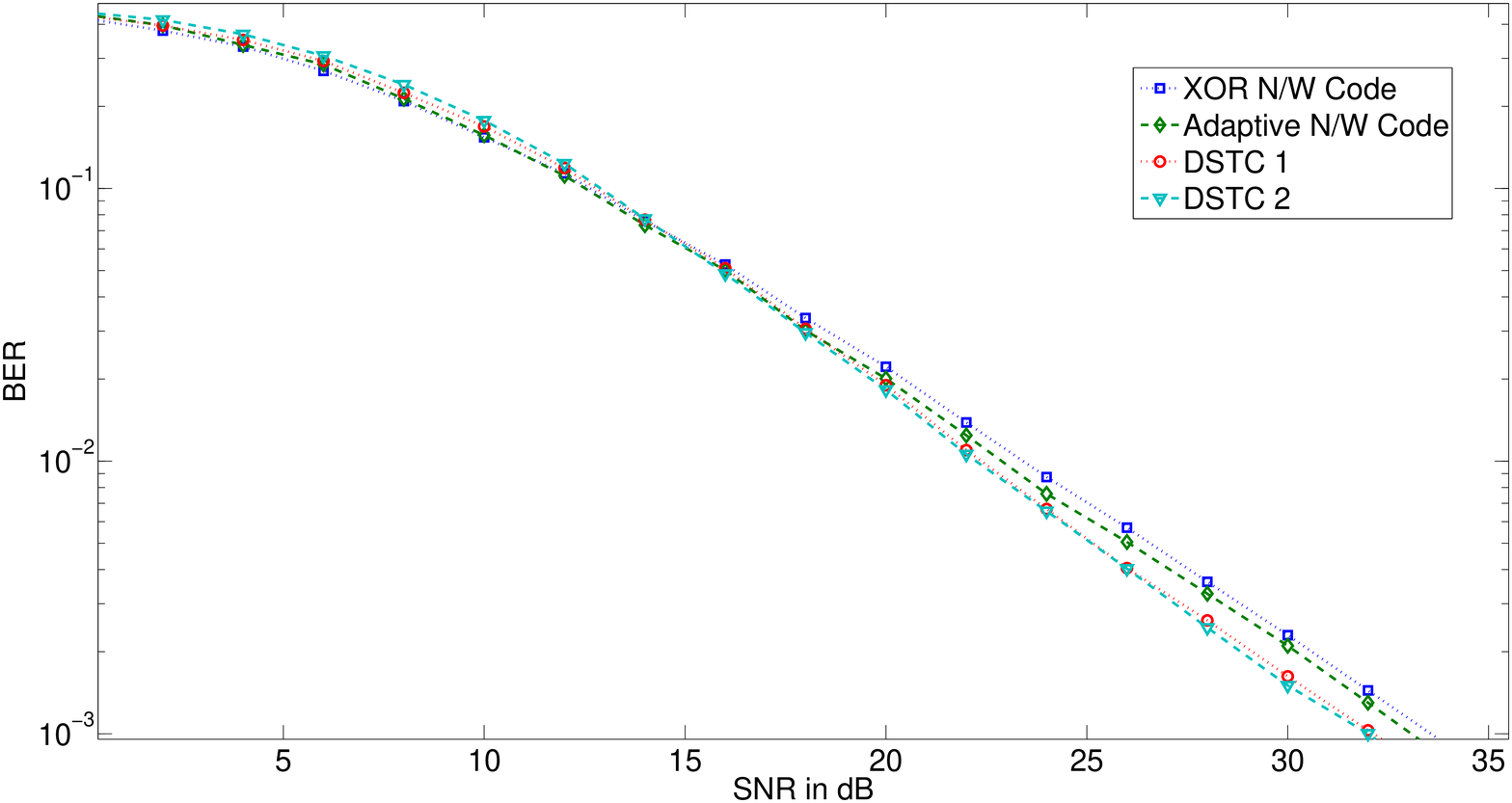}
\vspace{-.75 cm}
\caption{SNR vs BER for different schemes for 4-PSK signal set for a Rayleigh fading scenario.}	
\label{fig:ber_comp1}	
\end{figure*}
\begin{figure*}[htbp]
\centering
\vspace{-.2 cm}
\includegraphics[totalheight=4in,width=7in]{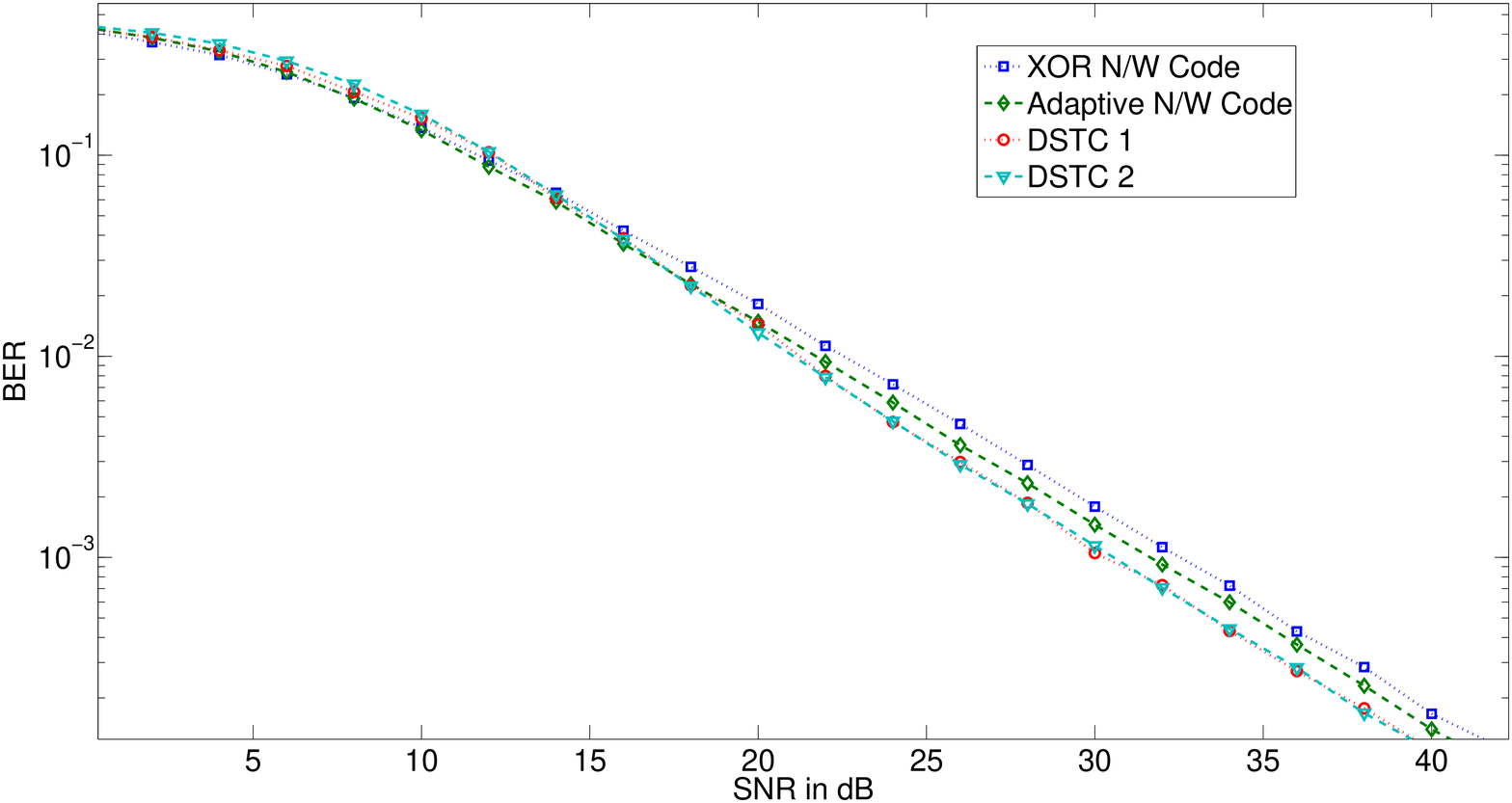}
\vspace{-.75 cm}
\caption{SNR vs BER for different schemes for 4-PSK signal set for a Rician fading scenario with a Rician factor 0 dB.}	
\label{fig:ber_comp2}	
\end{figure*}
\begin{figure*}[htbp]
\centering
\vspace{-.2 cm}
\includegraphics[totalheight=4in,width=7in]{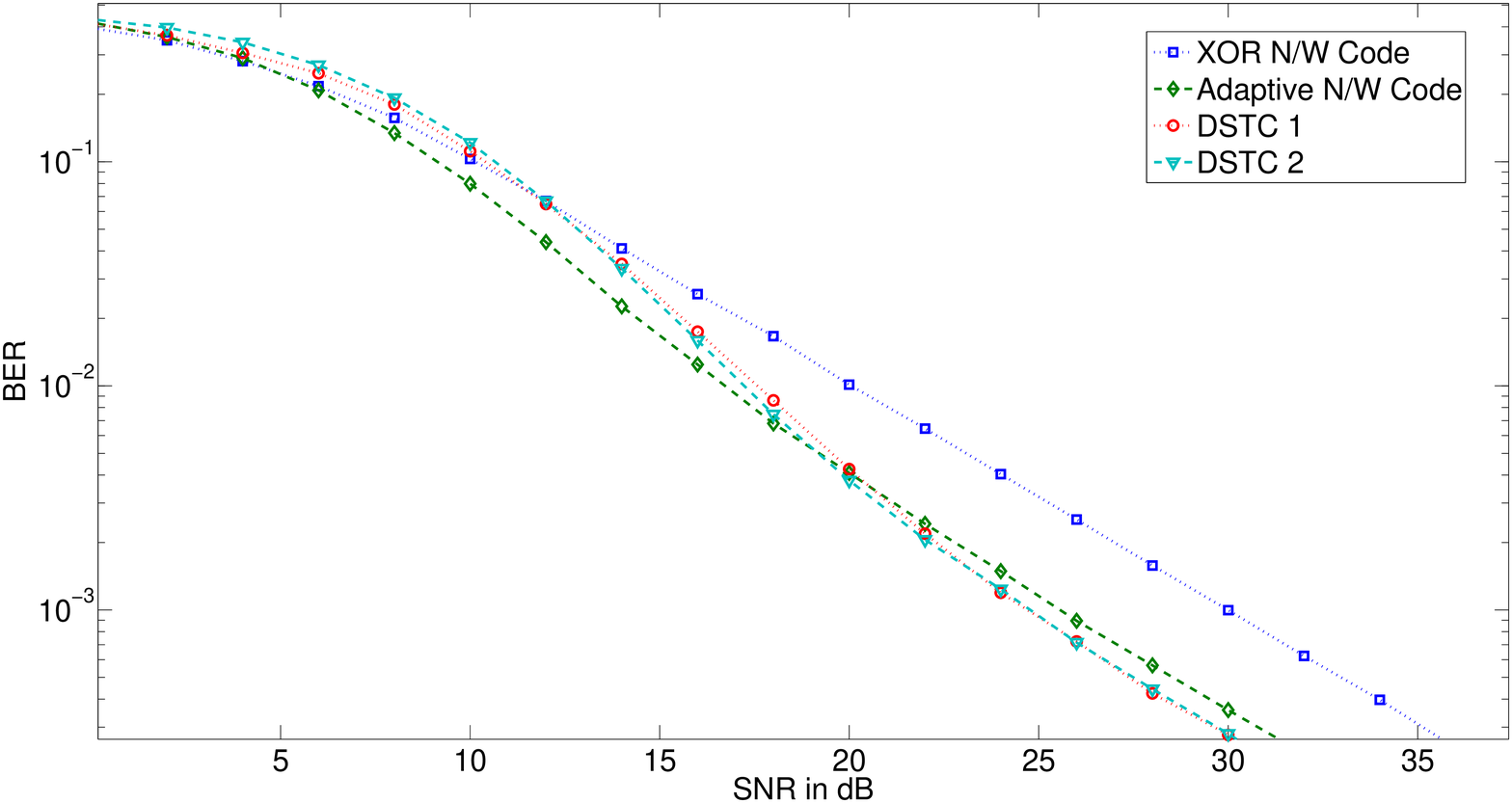}
\vspace{-.75 cm}
\caption{SNR vs BER for different schemes for 4-PSK signal set for a Rician fading scenario with a Rician factor 5 dB.}	
\label{fig:ber_comp3}	
\end{figure*}
\end{document}